\documentclass[hidelinks,onefignum,onetabnum,final]{siamart251216}

\usepackage{amsfonts}
\usepackage{graphicx}
\usepackage{epstopdf}
\usepackage{algorithmic}
\ifpdf
  \DeclareGraphicsExtensions{.eps,.pdf,.png,.jpg}
\else
  \DeclareGraphicsExtensions{.eps}
\fi

\newsiamremark{remark}{Remark}
\newsiamremark{hypothesis}{Hypothesis}
\crefname{hypothesis}{Hypothesis}{Hypotheses}
\newsiamthm{claim}{Claim}
\newsiamremark{fact}{Fact}
\crefname{fact}{Fact}{Facts}

\headers{Polarization-Induced Beam Bending}{H. Antil, R. L\"ohner, and S. Shah}

\title{Polarization-Induced Beam Bending:\\
Mathematical Model, Discretization, and Algorithm\thanks{Submitted to the editors DATE.
\funding{This work is also partially supported by the Office of Naval Research (ONR) under Award NO: N00014-24-1-2147,
NSF grant DMS-2408877, the Air Force Office of Scientific Research (AFOSR) under
Award NO: FA9550-25-1-0231, and SURE-AI Centre grant 357482, Research Council of Norway.}}}

\author{Harbir Antil\thanks{Center for Mathematics and Artificial Intelligence and Department of Mathematical Sciences, George Mason University,  Fairfax, Virginia 22030
  (\email{hantil@gmu.edu}, \email{sshah66@gmu.edu}).}
\and Rainald L\"ohner\thanks{Center for Computational Fluid Dynamics and Department of Physics and Astronomy, George Mason University, Fairfax, Virginia 22030. (\email{rlohner@gmu.edu})} 
\and Sarswati Shah\footnotemark[2]
}

\usepackage{amsopn}

\ifpdf
\hypersetup{
  pdftitle={Polarization-Induced Beam Bending:
Mathematical Model, Discretization, and Algorithm},
  pdfauthor={H. Antil, R. L\"ohner, S. Shah}
}
\fi

\begin{document}

\maketitle

\begin{abstract}
We study a reduced hydrodynamic formulation of paraxial vector beam propagation in which the beam intensity, optical phase, and spatially-dependent polarization are coupled through a nonlinear dispersive system. While prior analytical work derived a solution for the beam path valid for short propagation distances, a fully resolved numerical treatment of the model over long ranges has not previously been available.
 
Here we present a conservative numerical scheme for the coupled system, combining a finite-volume discretization of the intensity equation with monotone Hamilton--Jacobi (H-J) solvers for the phase dynamics and upwind transport of polarization. The method preserves the nonnegativity of the intensity and remains stable under long-distance propagation.

We perform large-scale simulations over propagation distances of tens of meters, while resolving millimeter-scale transverse structure. The numerical results reproduce the analytically predicted and experimentally observed quadratic beam
bending at short distances and reveal systematic deviations beyond the asymptotic regime. These deviations arise from nonlinear phase accumulation and dispersive effects captured by the full model but are neglected in the short-distance approximation.

\end{abstract}



\section{Introduction}
The propagation of paraxial optical vector beams, defined by their spatially varying polarization \cite{Chen:18}, admits a rich mathematical description that lies at the intersection of wave optics, fluid mechanics, and nonlinear partial differential equations. In recent years, reduced hydrodynamic formulations have emerged as a powerful
framework for understanding vector beam dynamics beyond purely geometric optics, while remaining tractable for analysis and computation \cite{Nichols:25a}.
These models describe the coupled evolution of beam intensity, phase, and
polarization-induced effects and provide a natural bridge between
microscopic wave equations and macroscopic transport phenomena. Such models have long been known in optics under the general heading of ``transport-of-intensity'' equations (see e.g., \cite{paganin1998noninterferometric,Zuo:15,nichols2018transport}) and are a natural application of the general hydrodynamic (Madelung-type) formulations of wave propagation (for details, see \cite{siegman1990new, madelung1927quantum}).

Hydrodynamic formulations of this type have been derived in several prior
works as asymptotic reductions of the paraxial wave equation \cite{siegman1990new}.
In particular, analytical studies have shown that spatial variations in
polarization can induce a systematic transverse bending of the beam,
leading to a curved trajectory of the beam centroid.
Closed-form expressions for this curvature have been obtained under
restrictive assumptions, such as weak polarization gradients, smooth
profiles, and short propagation distances \cite{JMNichols_DVNickel_FBucholtz_2022a}.
While these results provide valuable physical insight, they do not address
the fully nonlinear, coupled evolution of the intensity and phase fields
over long propagation ranges, nor do they quantify the regime of validity
of the asymptotic predictions. 

From a computational perspective, the numerical simulation of such
hydrodynamic beam models presents several challenges.
The phase equation is a nonlinear H-J equation that may
develop steep gradients, requiring monotone discretizations to capture
the correct viscosity solution.
The intensity equation is conservative and must preserve nonnegativity and
total mass.
The dispersion term involves second derivatives of the square-root of intensity and becomes singular in low-density regions. Finally, the polarization phase introduces additional coupling that must be handled consistently without destroying stability. To date, these difficulties have limited numerical studies primarily to simplified settings or short propagation distances.

The H-J equations are closely related to (scalar) conservation laws. A Godunov type upwind scheme for the H-J equation was proposed in \cite{kurganov2001semidiscrete}, a global Lax–Friedrichs flux splitting for H-J equations was introduced in \cite{osher1991high}, and a higher-order WENO scheme was presented in \cite{jiang2000weighted}. However, in the reduced hydrodynamic model, the Hamilton–Jacobi dynamics are coupled to the dispersive quantum-pressure term $Q(\rho)$, introducing a higher-order regularization that fundamentally changes the character of the phase equation and places additional stability demands on the discretization. The primary contribution of this work is the development and validation of a fully discrete numerical framework that enables stable, long-distance simulations of the complete reduced hydrodynamic model in two transverse dimensions. Our approach combines a conservative finite-volume discretization for the intensity equation \cite{leveque2002finite, jameson2008construction} with a monotone Godunov--Lax--Friedrichs scheme for the H-J phase equation.

Although LLF/Rusanov fluxes can be overly diffusive in purely hyperbolic settings \cite{edwards2006dominant}, in our model the dispersive term provides additional regularization; empirically, the LLF-based coupling yields stable long-distance propagation at the resolutions considered. To the best of our knowledge, a fully discrete coupling of a conservative finite-volume update for the intensity with a monotone Godunov–LLF discretization for the phase dynamics, tailored to the reduced hydrodynamic beam system with polarization forcing and quantum-pressure regularization, has not been previously reported and validated for stable, long-distance, fully resolved simulations. In this work, we derive an exact continuous energy balance (including boundary contributions). Reflecting boundary conditions are enforced in a manner consistent with mass conservation and the underlying physical model.

In addition to providing a detailed numerical methodology, we use the
resulting solver to investigate beam centroid dynamics over propagation
distances far beyond the regime accessible to existing analytical formulas.
This allows us to assess the range of validity of asymptotic curvature laws
and to identify qualitative changes in behavior that arise from fully
nonlinear coupling.
The numerical results thus complement and extend prior analytical work,
providing a more complete picture of polarization-induced beam bending. The short-distance bending law has also been compared against experiments in
\cite{JMNichols_DVNickel_FBucholtz_2022a}.

The remainder of the paper is organized as follows. 
In Section~\ref{sec:model}, we present the continuous model and discuss the role of each term. We also provide a proof of the momentum and energy balance. Section~\ref{sec:discretization} describes the numerical discretizations in detail, including stabilization and boundary treatment. In section~\ref{sec:discrete_results}, we establish key structural properties of the discrete scheme, including the exact conservation of mass and the positivity preservation of intensity. Section~\ref{sec:algorithm} discusses the pseudo-code and provides remarks on the choice of numerical schemes being used. Section~\ref{sec:numerics} presents numerical results and comparisons with theoretical and experimental predictions.
Concluding remarks and perspectives for future work are given in
Section~\ref{sec:outlook}.

% ==========================================================
\section{Mathematical Model}
\label{sec:model}

\subsection{Governing equations}

We consider a hydrodynamic model for the paraxial propagation of a monochromatic beam of light with wavenumber $k_0$.  The beam is presumed to be propagating predominantly in the $z>0$ direction with $(x,y)\in\Omega\subset\mathbb{R}^2$ denoting the transverse plane.  The electric field associated with the beam can be represented by the beam intensity $\rho(x,y,z)\ge 0$, a scalar dynamical phase $\phi(x,y,z)$
and a polarization-induced phase $\gamma(x,y,z)$.  The equations governing the evolution of these fields were described in \cite{Nichols:25a} and can be written
\begin{subequations}\label{eq:FormI}
\begin{align}
  \partial_z \rho
  + \frac{1}{k_0}\Div\big(\rho\,(\grad \phi + \grad \gamma)\big)
  &= 0,
  && \Omega \times (0,Z), \label{eq:rho}\\
  k_0\,\partial_z \phi
  + \frac12 |\grad\phi|^2
  + \frac12 |\grad\gamma|^2
  &= \frac12\,Q(\rho),
  && \Omega \times (0,Z), \label{eq:phi}\\
  \frac{D\gamma}{Dz} =   
  \partial_z \gamma
  + \frac{1}{k_0}\grad\phi\cdot\grad\gamma
  &= 0,
  && \Omega \times (0,Z). \label{eq:gamma}
\end{align}
\end{subequations}
The domain $\Omega$ is presumed bounded with a sufficiently smooth boundary $\partial\Omega$ and we let $Z>0$ denote the maximum propagation distance.
The model therefore couples a conservative transport equation for $\rho$ (\ref{eq:rho}) to a
Hamilton--Jacobi equation (H-J) for $\phi$ (\ref{eq:phi}), augmented by a dispersion-type
regularization (\ref{eq:gamma}).
Importantly, the polarization-induced phase $\gamma$ influences the phase dynamics through its gradient (last term on the left hand side of \ref{eq:phi}). The quantity
\begin{equation}\label{eq:Qdef}
  Q(\rho)
  := \frac{\Delta\sqrt{\rho}}{\sqrt{\rho}},
\end{equation}
is commonly referred to as a {\it quantum pressure}, so named for its appearance in hydro-dynamic models of the Schr\"odinger equation \cite{Nore:93}. In an optics context, this term is responsible for diffraction \cite{Nichols:25a}, acting in such a way as to drive the intensity profile toward a flatter distribution i.e. ``spread'' the beam out.  Gradient $\nabla$, Laplacian $\Delta$, and divergence $\Div$ operators are with respect to $(x,y) \in \Omega$. The system \eqref{eq:FormI} is supplemented with initial and boundary conditions,
which are specified below.

%------------------%
\subsection{Interpretation of the equations}
\paragraph{Intensity transport}
Equation~\eqref{eq:rho} expresses the (spatially) local conservation of intensity, transports the
intensity $\rho$ with velocity
$
  \bm v := \frac{1}{k_0}\grad \phi 
$.  Thus, the dynamical phase gradient is seen to define the optical path taken by each parcel of beam intensity \cite{Nichols:23}.  The argument of the divergence operator in \eqref{eq:rho} is the local momentum density which is seen to comprise of two components, $\rho\bm v$ and a term proportional to the polarization gradient 
${\bm \omega} := \frac{1}{k_0}\nabla \gamma.$  This latter component is unrelated to the optical path and serves as a ``stored'' momentum density $\rho\bm\omega$ available for exchange with the conventional momentum density $\rho\bm v$ during propagation \cite{Nichols:25a}.

\paragraph{Phase evolution}
Equation~\eqref{eq:phi} is a nonlinear H-J equation for the phase $\phi$. The quadratic term $|\grad\phi|^2$ drives the nonlinear evolution of $\phi$. The term $|\grad\gamma|^2$ acts as a polarization-induced forcing that couples polarization gradients into the phase dynamics and is the key mechanism behind polarization-driven bending in the coupled model.\;Lastly, the aforementioned diffraction term $Q(\rho)$ acts to spread the local intensity paths (as defined by the phase gradient) and prevents the formation of singularities in regions where $\rho$ remains
positive.

\paragraph{Polarization transport}
Equation~\eqref{eq:gamma} states that the polarization-induced phase is not changing along the local optical paths (total material derivative is zero) and is simply being advected by the velocity $\bm v$.

%------------------------%
\subsection{Boundary conditions and mass conservation}
We impose reflecting boundary conditions
\begin{equation}\label{eq:BC_neumann}
\partial_n \phi = \partial_n \gamma = 0 \quad \text{on }\partial\Omega\times(0,Z).
\end{equation}
However, we emphasize that for the application under consideration, the Neumann 
boundary conditions can be directly replaced by the decay at the boundary and all our
results below remain true. The reason
being that we choose the domain $\Omega$ to be sufficiently large so that the beam never 
reaches the boundary. 
Let $\theta=\phi+\gamma$ and recall $\bm v=k_0^{-1}\nabla\phi$ and $\bm \omega=k_0^{-1}\nabla\gamma$.
Then \eqref{eq:BC_neumann} implies ${\bm n}\cdot \bm v=0$ and ${\bm n}\cdot \bm \omega=0$ on $\partial\Omega$, and hence
\[
{\bm n}\cdot(\bm v+ \bm \omega)=\frac{1}{k_0}\,{\bm n}\cdot\nabla\theta=0 \qquad \text{on } \partial\Omega.
\] Consequently, the intensity equation
\eqref{eq:rho} satisfies the no-flux condition
\begin{equation}\label{eq:noflux_rho}
\bm n\cdot(\rho(\bm v + \bm \omega))=\frac{1}{k_0}\bm n\cdot(\rho\nabla\theta)=0
\quad\text{on }\partial\Omega\times(0,Z) ,
\end{equation}
ensuring that no intensity enters or leaves the computational domain. 
\begin{proposition}[Mass conservation]\label{prop:mass_continuous}
Under \eqref{eq:noflux_rho}, the total intensity
$
M(z):=\int_\Omega \rho(x,y,z)\,dx\,dy
$
is conserved: $M(z)=M(0)$ for all $z\in[0,Z]$.
\end{proposition}

\begin{proof}
Integrate \eqref{eq:rho} over $\Omega$ and use the divergence theorem:
\[
\frac{d}{dz}\int_\Omega \rho\,dxdy
= -\frac{1}{k_0}\int_\Omega \Div(\rho\nabla\theta)\,dxdy
= -\frac{1}{k_0}\int_{\partial\Omega} \bm n\cdot(\rho\nabla\theta)\,ds.
\]
The boundary integral vanishes by \eqref{eq:noflux_rho}; hence $dM/dz=0$.
\end{proof}

%--------------------------
\subsection{Momentum and energy balance}
This subsection provides the detailed proof of the momentum and energy balance laws, which hold for the system \eqref{eq:FormI}.
%------------------------------------------------------------
\paragraph{Local momentum equation (conservative form)}
%------------------------------------------------------------
Adding \eqref{eq:phi} and \eqref{eq:gamma} we get the $\theta-$H-J identity
\begin{equation}\label{eq:HJ_theta}
k_0 \partial_z \theta+\frac12|\nabla\theta|^2=\frac12Q(\rho).
\end{equation}
Here $\theta = \phi + \gamma$. Define the momentum density $m:=\frac{1}{k_0}\rho\,\nabla\theta \in \Omega.$ Then $(\rho,\theta)$ satisfy the \emph{local} momentum balance
\begin{equation}\label{eq:local_momentum_balance}
\partial_z m
+\Div\!\Big(\frac{1}{k_0^2}\,\rho\,\nabla\theta\otimes\nabla\theta + \bm S(\rho)\Big)
=0
\qquad\text{in }\Omega\times(0,Z).
\end{equation}
\textbf{Proof of \eqref{eq:local_momentum_balance}}. Differentiating $m=\frac1{k_0}\rho\nabla\theta$ yields
\[
\partial_z m=\frac1{k_0}\partial_z\rho\, \nabla\theta+\frac1{k_0}\rho\nabla\partial_z\theta.
\]
Use \eqref{eq:rho} to replace $\partial_z\rho$ and expand
$\Div(\rho\nabla\theta\otimes\nabla\theta)$ via the product rule:
\[
\Div(\rho\nabla\theta\otimes\nabla\theta)
=
\Div(\rho\nabla\theta)\,\nabla\theta
+\rho\,(\nabla\theta\cdot\nabla)\nabla\theta
=
\Div(\rho\nabla\theta)\,\nabla\theta
+\rho\,\nabla\!\Big(\tfrac12|\nabla\theta|^2\Big).
\]
Therefore
\[
\partial_z m + \Div\!\Big(\frac{1}{k_0^2}\rho\nabla\theta\otimes\nabla\theta\Big)
=
\frac1{k_0}\rho\nabla \partial_z\theta+\frac{1}{k_0^2}\rho\,\nabla\!\Big(\tfrac12|\nabla\theta|^2\Big).
\]
Next, take $\nabla$ of \eqref{eq:HJ_theta}:
\begin{equation} \label{eq:Nabla_HJ_theta}
k_0\nabla\partial_z \theta+\nabla\!\Big(\tfrac12|\nabla\theta|^2\Big)=\tfrac12\nabla Q(\rho),
\end{equation}
multiply by $\rho/k_0^2$ and substitute:
\[
\partial_z m + \Div\!\Big(\frac{1}{k_0^2}\rho\nabla\theta\otimes\nabla\theta\Big)
=
\frac{1}{2k_0^2}\rho\,\nabla Q(\rho).
\]
Finally, invoke Lemma~\ref{lem:Srelation_correct} in the form
$\Div \bm S(\rho)=-(1/(2k_0^2))\rho\nabla Q(\rho)$ to obtain \eqref{eq:local_momentum_balance}.

\begin{proposition}[Total momentum balance]\label{prop:momentum_balance_correct}
Assume $\rho\ge c_0>0$ and $(\rho,\phi,\gamma)$ are smooth solutions of \eqref{eq:FormI}. Define
\[
m:=\frac{1}{k_0}\rho\nabla\theta,
\qquad
\mathcal P(z):=\int_\Omega m(x,z)\,dx.
\]
Then $\mathcal P$ satisfies the exact boundary-flux identity
\begin{equation}\label{eq:global_momentum_flux}
\frac{d}{dz}\mathcal P(z)
=
-\int_{\partial\Omega}
\Big(\frac{1}{k_0^2}\rho\,\nabla\theta\otimes\nabla\theta+\bm S(\rho)\Big)\bm n\;ds.
\end{equation}
\end{proposition}
\begin{proof}
Integrate the local conservative law \eqref{eq:local_momentum_balance} over $\Omega$ and apply the divergence theorem:
\[
\frac{d}{dz}\int_\Omega m\,dx
+
\int_{\partial\Omega}
\Big(\frac{1}{k_0^2}\rho\,\nabla\theta\otimes\nabla\theta+\bm S(\rho)\Big)\bm n\;ds
=0.
\]
This is exactly \eqref{eq:global_momentum_flux}. 
\end{proof}

\begin{remark}[Momentum conservation] \label{remark:mom_conser}
The following holds
\begin{itemize}
\item On a bounded domain, momentum is \emph{not} conserved in general; it changes by the net momentum flux through $\partial\Omega$.
A sufficient condition for conservation is the \emph{stress-free} boundary condition (see \eqref{eq:global_momentum_flux})
\[
\Big(\frac{1}{k_0^2}\rho\,\nabla\theta\otimes\nabla\theta+\bm S(\rho)\Big)\bm n=0
\quad\text{on }\partial\Omega.
\]

\item On $\Omega=\mathbb R^2$ (with sufficient decay) or on a periodic torus, the right-hand side of \eqref{eq:global_momentum_flux} vanishes and $\mathcal P(z)$ is conserved.
\end{itemize}
\end{remark}

\begin{proposition}[Energy balance]
\label{prop:theta_energy_balance}
Assume there exists a constant $c_0>0$ such that $\rho(x,y,z)\ge c_0$ in $\Omega\times(0,Z)$, and $(\rho,\phi,\gamma)$ are smooth solutions of \eqref{eq:FormI}. Set $w:=\sqrt{\rho}$ and $Q(\rho):=\Delta w/w$.
Define the total energy
\[
\mathcal E(z)
:=
\frac{1}{2k_0^2}\int_\Omega \rho\,|\nabla\theta|^2\,dx
+\frac{1}{2k_0^2}\int_\Omega |\nabla w|^2\,dx.
\]
Assume the no-flux wall $\bm n\cdot(\rho\nabla\theta)=0$ on $\partial\Omega\times(0,Z)$.
Then, for every $z\in(0,Z)$,
\[
\frac{d}{dz}\mathcal E(z)
= \frac{1}{k_0^2}\int_{\partial\Omega} (\partial_n w)\,\partial_zw\,ds.
\]
Consequently, $\mathcal E(z)$ is conserved on $[0,Z]$ if the boundary contribution vanishes.
\end{proposition}

\begin{proof}
Fix $z\in(0,Z)$. We assume throughout that $\rho>0$ and $(\rho,\phi,\gamma)$ are smooth enough so that all derivatives and integrations by parts below are justified.

\medskip
\noindent\boxed{\mbox{\textrm Step 0:}} 
Let
\[
w:=\sqrt{\rho},\qquad Q(\rho):=\frac{\Delta w}{w} ,
\]
and 
\[
\mathcal E(z)
=
\underbrace{\frac{1}{2k_0^2}\int_\Omega \rho\,|\nabla\theta|^2\,dx}_{=: \mathcal E_{\theta,\mathrm{kin}}(z)}
\;+\;
\underbrace{\frac{1}{2k_0^2}\int_\Omega |\nabla w|^2\,dx}_{=: \mathcal E_{\theta,\mathrm{q}}(z)}.
\]

%------------------------------------------------------------
\medskip
\noindent\boxed{\mbox{\textrm Step 1:}} 
Differentiate $\mathcal E_{\theta,\mathrm{kin}}$ using the product rule:
\begin{align}
\frac{d}{dz}\mathcal E_{\theta,\mathrm{kin}}
&=
\frac{1}{2k_0^2}\int_\Omega \partial_z\rho\,|\nabla\theta|^2\,dx
+\frac{1}{k_0^2}\int_\Omega \rho\,\nabla\theta\cdot\nabla\partial_z\theta\,dx.
\label{eq:dEkin_theta_split_pf}
\end{align}

\smallskip
\noindent{\it Step 1a: eliminate $\nabla\partial_z\theta$ via the $\theta-$H-J equation.}
Dot equation \eqref{eq:Nabla_HJ_theta} with $\rho\nabla\theta$ and integrate over $\Omega$:
\begin{align}
k_0\int_\Omega \rho\,\nabla\theta\cdot\nabla\partial_z\theta\,dx
+ \int_\Omega \rho\,\nabla\theta\cdot \nabla\!\Big(\tfrac12|\nabla\theta|^2\Big)\,dx
&=
\frac12\int_\Omega \rho\,\nabla\theta\cdot\nabla Q(\rho)\,dx.
\label{eq:HJ_grad_dotted_pf}
\end{align}
Divide by $k_0^3$ and substitute into the second term in \eqref{eq:dEkin_theta_split_pf}:
\begin{align}
\frac{1}{k_0^2}\int_\Omega \rho\,\nabla\theta\cdot\nabla\partial_z\theta\,dx
&=
\frac{1}{2k_0^3}\int_\Omega \rho\,\nabla\theta\cdot\nabla Q(\rho)\,dx
-\frac{1}{k_0^3}\int_\Omega \rho\,\nabla\theta\cdot \nabla\!\Big(\tfrac12|\nabla\theta|^2\Big)\,dx.
\label{eq:term_theta_grad_theta_z_pf}
\end{align}
Plugging \eqref{eq:term_theta_grad_theta_z_pf} into \eqref{eq:dEkin_theta_split_pf} yields
\begin{align}
\begin{aligned}
\frac{d}{dz}\mathcal E_{\theta,\mathrm{kin}}
&=
\frac{1}{2k_0^2}\int_\Omega \partial_z\rho\,|\nabla\theta|^2\,dx
+\frac{1}{2k_0^3}\int_\Omega \rho\,\nabla\theta\cdot\nabla Q(\rho)\,dx \\
&\quad-\frac{1}{k_0^3}\int_\Omega \rho\,\nabla\theta\cdot \nabla\!\Big(\tfrac12|\nabla\theta|^2\Big)\,dx.
\label{eq:dEkin_theta_three_terms_pf}
\end{aligned}
\end{align}

\smallskip
\noindent{\it Step 1b: convert the last two terms using the continuity equation.}
Let $f$ be any smooth scalar field (depending on $x$ and the fixed $z$). 
The standard Green's identity with $\bm n \cdot \nabla \theta =0$ gives:
\begin{equation}\label{eq:key_flux_identity_pf_noflux_correct}
\int_\Omega \rho\,\nabla\theta\cdot\nabla f\,dx
=
-\int_\Omega f\,\Div(\rho\nabla\theta)\,dx. \nonumber
\end{equation}
Then using \eqref{eq:FormI} gives $\Div(\rho\nabla\theta)=k_0(-\partial_z\rho)$, and therefore
\begin{equation}\label{eq:rho_theta_convert_pf_correct}
\int_\Omega \rho\,\nabla\theta\cdot\nabla f\,dx
= k_0\int_\Omega \partial_z\rho\,f\,dx .
\end{equation}

Apply \eqref{eq:rho_theta_convert_pf_correct} with $f=Q(\rho)$ and $f=\tfrac12|\nabla\theta|^2$ to obtain
\begin{align}
\int_\Omega \rho\,\nabla\theta\cdot\nabla Q(\rho)\,dx
&=
k_0\int_\Omega \partial_z\rho\,Q(\rho)\,dx,
\label{eq:convert_Q_pf_correct}\\
\int_\Omega \rho\,\nabla\theta\cdot\nabla\!\Big(\tfrac12|\nabla\theta|^2\Big)\,dx
&=
k_0\int_\Omega \partial_z\rho\,(\tfrac12|\nabla\theta|^2)\,dx.
\label{eq:convert_theta_sq_pf_correct}
\end{align}
Insert \eqref{eq:convert_Q_pf_correct}--\eqref{eq:convert_theta_sq_pf_correct} into
\eqref{eq:dEkin_theta_three_terms_pf}:
\begin{align}
\frac{d}{dz}\mathcal E_{\theta,\mathrm{kin}}
&=
\frac{1}{2k_0^2}\int_\Omega \partial_z\rho\,|\nabla\theta|^2\,dx
+\frac{1}{2k_0^3}\,k_0\int_\Omega \partial_z\rho\,Q(\rho)\,dx
-\frac{1}{k_0^3}\,k_0\int_\Omega \partial_z\rho\,(\tfrac12|\nabla\theta|^2)\,dx
\nonumber\\
&=
\frac{1}{2k_0^2}\int_\Omega \partial_z\rho\,Q(\rho)\,dx .
\label{eq:dEkin_theta_final_correct}
\end{align}

%------------------------------------------------------------
\medskip
\noindent\boxed{\mbox{\textrm Step 2:}} Differentiate  the ``quantum'' part $\mathcal E_{\theta,\mathrm{q}}$:
\begin{align}
\frac{d}{dz}\mathcal E_{\theta,\mathrm{q}}
&=
\frac{1}{2k_0^2}\frac{d}{dz}\int_\Omega |\nabla w|^2\,dx
=
\frac{1}{k_0^2}\int_\Omega \nabla w\cdot\nabla \partial_zw\,dx.
\label{eq:dEq_start_pf}
\end{align}
Integrate by parts in space:
\[
\int_\Omega \nabla w\cdot\nabla \partial_zw\,dx
=
-\int_\Omega (\Delta w)\,\partial_zw\,dx
+\int_{\partial\Omega} (\partial_n w)\,\partial_zw\,ds.
\]
Since $\rho=w^2$, we have $\partial_z\rho=2w\partial_zw$, i.e.,\ $\partial_zw=\partial_z\rho/(2w)$, and therefore
\[
(\Delta w)\,\partial_zw
=
(\Delta w)\frac{\partial_z\rho}{2w}
=
\frac12\,\partial_z\rho\,\frac{\Delta w}{w}
=
\frac12\,\partial_z\rho\,Q(\rho).
\]
Substitute this into \eqref{eq:dEq_start_pf} to obtain
\begin{equation}\label{eq:dEq_final_pf}
\frac{d}{dz}\mathcal E_{\theta,\mathrm{q}}
= -\frac{1}{2k_0^2}\int_\Omega \partial_z\rho\,Q(\rho)\,dx + \frac{1}{k_0^2}\int_{\partial\Omega} (\partial_n w)\,\partial_zw\,ds.
\end{equation}
Adding \eqref{eq:dEkin_theta_final_correct} and \eqref{eq:dEq_final_pf}, the energy $\mathcal E(z)$ satisfies the exact boundary-balance identity above.
\end{proof}

% ==========================================================
\subsection{Centroid identity}
\label{subsec:S3_short_distance}
The objective of this subsection is to show that, for the full model \eqref{eq:FormI}, the centroid does not drift because momentum is conserved. {This is indeed the case for the class of solutions we are considering. Given the nonlinear (possible solution nonuniqueness) nature of the problem, it may be possible to identify other solutions exhibiting a different behavior.} We begin by deriving an exact expression for the centroid associated with the full model. Throughout this section, assume $\rho_0=\rho(\cdot,0)$ is smooth and strictly positive on $\Omega$, $\phi_0=\phi(\cdot,0)$, and $\gamma_0=\gamma(\cdot,0)$ are smooth, and Neumann conditions \eqref{eq:BC_neumann} hold. We first derive the following centroid evolution identity.

% ----------------------------------------------------------
\begin{lemma}[Centroid evolution identity]\label{lem:centroid_identity}
Let $
0 < M(z):=\int_\Omega \rho(x,y,z)\,dx\,dy
< \infty$, and \eqref{eq:noflux_rho}. Define the centroid
\begin{equation}\label{eq:centroid}
(x_c(z),y_c(z))
:=\frac{1}{M(z)}\int_\Omega (x,y)\,\rho(x,y,z)\,dx\,dy.
\end{equation}
Then
\[
\frac{d}{dz}x_c(z)=\frac{1}{k_0 M(z)}\int_\Omega \rho\,\partial_x\theta\,dx\,dy,
\qquad
\frac{d}{dz}y_c(z)=\frac{1}{k_0 M(z)}\int_\Omega \rho\,\partial_y\theta\,dx\,dy.
\]
\end{lemma}

\begin{proof}
We show the $y$-identity; the $x$-case is identical. Differentiate under the
integral and use \eqref{eq:rho}:
\[
\frac{d}{dz}\int_\Omega y\rho\,dxdy
= \int_\Omega y\,\partial_z\rho\,dxdy
= -\frac{1}{k_0}\int_\Omega y\,\Div(\rho\nabla\theta)\,dxdy.
\]
Integrate by parts:
\[
-\int_\Omega y\,\Div(\rho\nabla\theta)\,dxdy
= -\int_{\partial\Omega} y\,\bm n\cdot(\rho\nabla\theta)\,ds
+ \int_\Omega \nabla y\cdot(\rho\nabla\theta)\,dxdy.
\]
The boundary term vanishes by \eqref{eq:noflux_rho}; since $\nabla y=(0,1)$,
the volume term becomes $\int_\Omega \rho\,\partial_y\theta\,dxdy$. Divide by
$M(z)$ and use Proposition~\ref{prop:mass_continuous}, i.e., $\tfrac{dM}{dz} = 0$ 
to conclude.
\end{proof}

%%---------------------%%
\begin{theorem}[Linear-in-$z$ order]
\label{thm:yc_linear}\;Let the solution $(\rho,\phi,\gamma)$ to \eqref{eq:FormI} be sufficiently smooth. Then
\begin{equation}\label{eq:ycz}
y_c(z) = y_c(0)+\frac{\mathcal P_y}{M(z)}z,
\end{equation}
where $\mathcal P(z) = (\mathcal P_x(z), \mathcal P_y(z))$ is total transverse momentum defined in Proposition~\ref{prop:momentum_balance_correct}.
\end{theorem}
\begin{proof}
Lemma~\ref{lem:centroid_identity} gives
$y_c'(z)=\frac{1}{k_0 M(z)}\int_\Omega \rho\,\partial_y\theta \,dx\,dy$. Using the definition of $\mathcal P(z)$, we arrive at (here ${\mathcal P_y}$ is the $y$-th component)
\[
y_c'(z) = \frac{\mathcal P_y(z)}{M(z)}.
\]
In view of Proposition~\ref{prop:momentum_balance_correct} and Remark~\ref{remark:mom_conser}, 
the total momentum vector $\mathcal P(z)$ is conserved, its $y$-th component is also conserved. 
Consequently, $y_c'(z) = \frac{\mathcal P_y}{M(z)}$ is constant, and therefore \eqref{eq:ycz}
holds.
\end{proof} 

This derivation uses the original three-equation system directly. The key point is that the apparent forcing from \( |\grad\gamma|^2\) in the \(\phi\)-equation is exactly balanced by the evolution equation for \(\gamma\).

%%---------------------------%%
\subsection{Reduced model}
The experimental observation and short-distance theory in \cite{Nichols:25a} motivates us to also focus on the \emph{modified} (reduced) intensity equation,
which is free of the $\nabla\gamma$ contribution. Consequently, the model becomes:
\begin{subequations}\label{eq:Form_mod}
\begin{align}
  \partial_z \rho
  + \frac{1}{k_0}\Div\big(\rho\,\grad \phi \big)
  &= 0,
  && \Omega \times (0,Z), \label{eq:rho_mod}\\
  k_0\,\partial_z \phi
  + \frac12 |\grad\phi|^2
  + \frac12 |\grad\gamma|^2
  &= \frac12\,Q(\rho),
  && \Omega \times (0,Z), \label{eq:phi_nmod}\\
  \frac{D\gamma}{Dz} =   
  \partial_z \gamma
  + \frac{1}{k_0}\grad\phi\cdot\grad\gamma
  &= 0,
  && \Omega \times (0,Z). \label{eq:gamma_nmod}
\end{align}
\end{subequations}

\begin{theorem}[Nonlinear-in-$z$ order] \label{thm:yc_nonlinear}
Let the solution $(\rho,\phi,\gamma)$ of \eqref{eq:Form_mod} be sufficiently smooth
and $\int_{\partial\Omega} \bm S(\rho) \,\bm n \, ds = 0$. Then the
transverse momentum
${\mathcal P}(z):=\frac{1}{k_0}\int_\Omega \rho \nabla \theta \, dx\,dy,$
satisfies 
\[
{\mathcal P}'(z)
=
-\frac{1}{k_0^2}\int_\Omega \rho\, D^2\theta\, \nabla\gamma \, dx\,dy.
\]
Here, $D^2\theta$ is the Hessian of $\theta$. In particular, the $y$-component satisfies
\[
{\mathcal P}_y'(z)
=
-\frac{1}{k_0^2}\int_\Omega \rho\,\nabla\gamma\cdot \nabla(\partial_y\theta)\,dx\,dy.
\]
The centroid $y_c(z)$ can be expressed as
\[
y_c(z)
=
y_c(0)
+\frac{{\mathcal P}_y(0)}{M}\,z
-\frac{1}{M}\int_0^z \int_0^s F_y(\tau)\,d\tau\,ds
-\frac{1}{k_0 M}\int_0^z\left(\int_\Omega \rho\,\gamma_y\,dx\,dy\right)ds,
\]
where
\[
F_y(z)
:=
\frac{1}{k_0^2}\int_\Omega \rho\,\nabla\gamma\cdot \nabla(\partial_y\theta)\,dx\,dy.
\]
\end{theorem}

\begin{proof}
Total transverse momentum for the reduced system \eqref{eq:Form_mod} is the same 
as before (see also \cite{Nichols:25a}):
${\mathcal P}(z):= \frac{1}{k_0}\int_{\Omega}\rho(\grad \phi+ \grad \gamma)\,dx\,dy = \frac{1}{k_0}\int_{\Omega}\rho \grad \theta\,dx\,dy $.
Using \eqref{eq:rho_mod} and Lemma~\ref{lem:centroid_identity}, we obtain
\begin{equation}\label{eq:centroid_prime}
y_c'(z) = \frac{1}{k_0 M}\int_\Omega \rho\,\partial_y\phi\,dx\,dy =  \frac{1}{M(z)} \left({\mathcal P}_y(z) - \frac{1}{k_0}\int_\Omega \rho \gamma_y  \,dxdy \right).
\end{equation}
 Next, we show that ${\mathcal P}(z)$ is no longer conserved. Differentiating ${\mathcal P}(z)$
and using \eqref{eq:rho_mod} and \eqref{eq:Nabla_HJ_theta}, we get
\[
 {\mathcal P}'(z)
=
-\frac{1}{k_0^2}\int_\Omega \Div(\rho\nabla\phi)\,\nabla\theta 
-\frac{1}{k_0^2}\int_\Omega \rho\,\nabla\!\left(\frac12 |\nabla\theta|^2\right) 
+\frac{1}{2k_0^2}\int_\Omega \rho\,\nabla Q(\rho)\,.
\]
Integrating the first term by parts, using the boundary condition, and invoking Lemma~\ref{lem:Srelation_correct} in the form
$\Div \bm S(\rho)=-(1/(2k_0^2))\rho\nabla Q(\rho)$, we obtain
\[
{\mathcal P}'(z)
=
-\frac{1}{k_0^2}\int_\Omega \rho\,\,D^2\theta\,\nabla\gamma\,dx\,dy
-\int_\Omega \Div \bm S(\rho) \,dx\,dy.
\]
Using the divergence theorem, the above equation becomes
\begin{equation} \label{eq:mome_prime}
{\mathcal P}'(z)
= -\frac{1}{k_0^2}\int_\Omega \rho\,\,D^2\theta\,\nabla\gamma\,dx\,dy
- \int_{\partial \Omega} \bm S(\rho)\, \bm n \, ds.
\end{equation}
In view of our assumption, the boundary term vanishes. 
Hence, only the first term remains. Taking the $y$-component, we obtain
\[
{\mathcal P}_y'(z)
=
-\frac{1}{k_0^2}\int_\Omega \rho\,\nabla\gamma\cdot \nabla(\partial_y\theta)\,dx\,dy
=: -F_y(z).
\]
Thus,
\[
{\mathcal P}_y(z)={\mathcal P}_y(0)-\int_0^z F_y(s)\,ds.
\]
Substituting this into \eqref{eq:centroid_prime} and integrating from $0$ to $z$, we obtain
\[
y_c(z)
=
y_c(0)
+\frac{{\mathcal P}_y(0)}{M}\,z
-\frac{1}{M}\int_0^z \int_0^s F_y(\tau)\,d\tau\,ds
-\frac{1}{k_0 M}\int_0^z\left(\int_\Omega \rho\,\gamma_y\,dx\,dy\right)ds.
\label{2.25}
\]
This shows that, in general, the centroid is nonlinear in $z$.
\end{proof}

% ==========================================================
\section{Fully Specified Numerical Scheme}
\label{sec:discretization}

Next, we discuss the discretization and implementation of a numerical
scheme to solve \eqref{eq:FormI}. 

% ----------------------------------------------------------
\subsection{Grid, cells, and face indexing}
\label{subsec:grid}

All unknowns are discretized on a uniform Cartesian grid
$\{(x_i,y_j)\}_{i=1,j=1}^{N_x,N_y}$ covering $\Omega=(-L,L)^2$:
\[
  x_i = -L + \Big(i-\tfrac12\Big)\Delta x,\quad i=1,\dots,N_x,
  \qquad
  y_j = -L + \Big(j-\tfrac12\Big)\Delta y,\quad j=1,\dots,N_y,
\]
with $\Delta x = 2L/N_x$ and $\Delta y = 2L/N_y$.
We store cell-centered unknowns
$\rho_{i,j}\approx\rho(x_i,y_j,z)$, $\phi_{i,j}$, and $\gamma_{i,j}$.

For finite-volume fluxes, we use standard face notation.
The cell $(i,j)$ corresponds to
$[x_{i-\frac12},x_{i+\frac12}]\times[y_{j-\frac12},y_{j+\frac12}]$.
The face index $(i+\tfrac12,j)$ denotes the interface between cells
$(i,j)$ and $(i+1,j)$, which lies on the vertical line $x=x_{i+\frac12}$
(hence a \emph{vertical} face).
Analogously, $(i,j+\tfrac12)$ denotes a \emph{horizontal} face.

% ----------------------------------------------------------
\subsection{Wall (Neumann) boundary conditions via ghost cells}
\label{subsec:ghost}

Reflecting (homogeneous Neumann) boundary conditions for $\phi$ and $\gamma$
are imposed using ghost-cell reflection so that centered finite differences
can be applied up to the boundary.
For any grid function $u\in\R^{N_x\times N_y}$ we define an extended array
$u^e\in\R^{(N_x+2)\times(N_y+2)}$ by
\[
  u^e_{i+1,j+1} = u_{i,j},
  \qquad 1\le i\le N_x,\ 1\le j\le N_y,
\]
and mirror interior values into ghost cells, e.g.
\[
  u^e_{1,j+1} = u_{2,j},
  \qquad
  u^e_{N_x+2,j+1} = u_{N_x-1,j},
\]
(Equivalently, $u^e_{0,j}=u^e_{2,j}$ and $u^e_{N_x+1,j}=u^e_{N_x-1,j}$ in zero-based ghost notation.) 
and similarly in the $y$-direction.
This enforces $\partial_n u=0$ at $\partial\Omega$ in a discrete sense.

Centered first derivatives are then defined by
\begin{align*}
  (D_x^0 u)_{i,j}
  &:= \frac{u^e_{i+2,j+1}-u^e_{i,j+1}}{2\Delta x},\qquad 
  (D_y^0 u)_{i,j}
  := \frac{u^e_{i+1,j+2}-u^e_{i+1,j}}{2\Delta y},
\end{align*}
and the discrete Laplacian by
\[
  (\Delta_h u)_{i,j}
  :=
  \frac{u^e_{i+2,j+1}-2u^e_{i+1,j+1}+u^e_{i,j+1}}{\Delta x^2}
  +
  \frac{u^e_{i+1,j+2}-2u^e_{i+1,j+1}+u^e_{i+1,j}}{\Delta y^2}.
\]

% ----------------------------------------------------------
\subsection{Dispersive term $Q(\rho)$ and near-vacuum regularization}
\label{subsec:Q}

We approximate $Q(\rho)$ from \eqref{eq:Qdef} as follows.
First enforce a positivity floor $\rho\leftarrow\max(\rho,\rho_{\min})$.
Define $s_{i,j}=\sqrt{\rho_{i,j}}$ and compute $(\Delta_h s)_{i,j}$.
To avoid division by very small values we use
\[
  s^{\mathrm{den}}_{i,j}=\max(s_{i,j},s_{\min}),
  \qquad
  Q_{i,j}=\frac{(\Delta_h s)_{i,j}}{s^{\mathrm{den}}_{i,j}}.
\]
However, in our numerical experiments we set $\rho_{\rm min} = 10^{-20}\max\rho_0$ and 
$s_{\rm min} = 0$.

% ----------------------------------------------------------
\subsection{Conservative discretization of the intensity equation}
\label{subsec:rho}

The intensity equation \eqref{eq:rho}
is discretized in a conservative form
\[
  \partial_z\rho_{i,j} =
  -\frac{F^x_{i+\frac12,j}-F^x_{i-\frac12,j}}{\Delta x}
  -\frac{F^y_{i,j+\frac12}-F^y_{i,j-\frac12}}{\Delta y}.
\]

\paragraph{Rusanov (local Lax--Friedrichs) flux}
At each interior vertical face $(i+\tfrac12,j)$, define a face velocity by
averaging neighboring cell-centered velocities:
\[
  v^x_{i,j}:=\frac{1}{k_0}(D_x^0\phi)_{i,j},
  \qquad
  v^x_{i+\frac12,j}:=\tfrac12\big(v^x_{i,j}+v^x_{i+1,j}\big),
  \qquad
  a^x_{i+\frac12,j}:=\big|v^x_{i+\frac12,j}\big|,
\]
and the polarization gradient
\[
  \omega^x_{i,j}:=\frac{1}{k_0}(D_x^0\gamma)_{i,j},
  \qquad
  \omega^x_{i+\frac12,j}:=\tfrac12\big(\omega^x_{i,j}+\omega^x_{i+1,j}\big),
  \qquad
  b^x_{i+\frac12,j}:=\big|\omega^x_{i+\frac12,j}\big|.
\]
The Rusanov flux is
\[
  F^x_{i+\frac12,j}
  =
  \tfrac12\big(\rho_{i,j}+\rho_{i+1,j}\big)\left(v^x_{i+\frac12,j} + \omega^x_{i+\frac12,j} \right)
  -\tfrac12 \left( a^x_{i+\frac12,j} + b^x_{i+\frac12,j} \right)\big(\rho_{i+1,j}-\rho_{i,j}\big),
\]
and the $y$-flux is defined analogously.

\paragraph{No-flux boundary}
To enforce the reflecting wall condition, we set the boundary face fluxes to zero:
\[
  F^x_{\frac12,j}=F^x_{N_x+\frac12,j}=0,\qquad
  F^y_{i,\frac12}=F^y_{i,N_y+\frac12}=0,
\]
which enforces $\bm n\cdot(\rho\grad \theta)=0$ discretely and yields conservation of total mass up to $z$-integration error.

\paragraph{Positivity floor}
After each Runge--Kutta stage ($z-$ integrator, described below) we apply $\rho\leftarrow\max(\rho,\rho_{\min})$.

% ----------------------------------------------------------
\subsection{Monotone H-J discretization of the phase equation}
\label{subsec:phi}

The phase equation \eqref{eq:phi} is written in H-J form
\begin{equation} \label{eq:HJ_Hamiltonian}
\partial_z\phi + H(\nabla\phi)=S,
\qquad
H(p)=\frac{1}{2k_0}|p|^2,\quad
S=\frac{1}{k_0}\Big(\tfrac12 Q(\rho)-\tfrac12|\nabla\gamma|^2\Big).
\end{equation}
{Here, we write the Hamiltonian in terms of the transverse phase gradient $p = \nabla \phi = (p_x, p_y)$, where $p_x = \partial_x\phi$ and $p_y = \partial_y \phi.$} Nonlinear H-J equations may develop steep gradients even from smooth initial data; stable computation, therefore, requires a monotone scheme that converges to the viscosity solution \cite{crandal1984two, souganidis1985approximation}.

\paragraph{One-sided slopes}
Using ghost extensions, define
\[
  (D_x^+ \phi)_{i,j} = \frac{\phi^e_{i+2,j+1}-\phi^e_{i+1,j+1}}{\Delta x},
  \qquad
  (D_x^- \phi)_{i,j} = \frac{\phi^e_{i+1,j+1}-\phi^e_{i,j+1}}{\Delta x},
\]
and similarly $(D_y^\pm\phi)_{i,j}$.

\paragraph{Godunov Hamiltonian and LLF stabilization}
For the convex quadratic Hamiltonian $H(p)=\frac{1}{2k_0}(p_x^2+p_y^2)$, we use
the Godunov numerical Hamiltonian 
\begin{align}
&H_G(\phi)_{i,j} \label{eq:GodunovHamiltonia}\\
& \ =
\tfrac{1}{2k_0}\Big(
\max(D_x^-\phi_{i,j},0)^2 + \min(D_x^+\phi_{i,j},0)^2
+
\max(D_y^-\phi_{i,j},0)^2 + \min(D_y^+\phi_{i,j},0)^2
\Big). \nonumber
\end{align}
To improve robustness, we add local Lax--Friedrichs dissipation with coefficients chosen to bound the characteristic speeds:
\begin{equation} \label{eq:charspeed}
\alpha_x = \frac{1}{k_0}\max_{i,j}\left(|(D_x^-\phi)_{i,j}|, |(D_x^+\phi)_{i,j}|\right), \,\,\,\,\,
\alpha_y = \frac{1}{k_0}\max_{i,j} \left(|(D_y^-\phi)_{i,j}|, |(D_y^+\phi)_{i,j}|\right),
\end{equation}
and 
\begin{equation} \label{eq:diss}
\mathrm{diss}_{i,j}
=
\frac12\alpha_x\big((D_x^+\phi)_{i,j}-(D_x^-\phi)_{i,j}\big)
+
\frac12\alpha_y\big((D_y^+\phi)_{i,j}-(D_y^-\phi)_{i,j}\big).
\end{equation}
We then approximate $H(\nabla\phi)\approx H_G-\mathrm{diss}$.

\paragraph{Polarization forcing}
Since $\gamma$ satisfies a linear transport equation, its gradient remains
smooth in our setting.
We therefore compute
\[
  |\nabla\gamma|^2_{i,j} = (D_x^0\gamma_{i,j})^2 + (D_y^0\gamma_{i,j})^2,
\]
using centered differences.

\paragraph{Discrete phase RHS}
The semi-discrete phase update is
\[
  (\partial_z\phi)_{i,j}
  =
  \frac{1}{k_0}\Big(\tfrac12 Q_{i,j} - \tfrac12 |\nabla\gamma|^2_{i,j}\Big)
  - \big(H_G(\phi)_{i,j}-\mathrm{diss}_{i,j}\big).
\]

% ----------------------------------------------------------
\subsection{Upwind discretization of the $\gamma$-equation}
\label{subsec:gamma}

The polarization equation is
\[
  \partial_z\gamma + \bm v\cdot\nabla\gamma=0,
  \qquad
  \bm v=\frac{1}{k_0}\nabla\phi.
\]
We discretize it by directional upwinding:
\[
  (\partial_z\gamma)_{i,j} =
  -(v_x)_{i,j}(D_x^{\rm up}\gamma)_{i,j}
  -(v_y)_{i,j}(D_y^{\rm up}\gamma)_{i,j},
\]
with
\[
  (D_x^{\rm up}\gamma)_{i,j} =
  \begin{cases}
    (D_x^-\gamma)_{i,j}, & (v_x)_{i,j}\ge 0,\\
    (D_x^+\gamma)_{i,j}, & (v_x)_{i,j}< 0,
  \end{cases}
\]
and analogously for $D_y^{\rm up}$.

% ----------------------------------------------------------
\subsection{SSP-RK3 $z$ stepping and CFL control}
\label{subsec:time}

Let $U=(\rho,\phi,\gamma)$ and let $\mathcal{R}(U)$ denote the semi-discrete
right-hand side defined by the operators above.
We advance the discrete form of \eqref{eq:FormI} using the third-order SSP-RK3 scheme:
\begin{equation}\label{eq:SSP-RK3}
\begin{aligned}
  U^{(1)} &= U^n + \Delta z\,\mathcal{R}(U^n),\\
  U^{(2)} &= \frac34 U^n + \frac14\Big(U^{(1)} + \Delta z\,\mathcal{R}(U^{(1)})\Big),\\
  U^{n+1} &= \frac13 U^n + \frac23\Big(U^{(2)} + \Delta z\,\mathcal{R}(U^{(2)})\Big).
\end{aligned}
\end{equation}

We additionally enforce 
\[
  \phi \leftarrow \phi - \mathrm{mean}(\phi),
\]
after each stage, which fixes the additive gauge invariance of $\phi$ and improves
numerical conditioning without affecting $\nabla\phi$.

\paragraph{CFL step size} Since all three updates are explicit, we chose $\Delta z$ to satisfy a Courant--Friedrichs--Lewy (CFL) restriction based on the maximum characteristic transport speeds. 

Using the wave-speed and polarization  parameters $a^x_{i+\frac12,j},\,b^x_{i+\frac12,j}$ and $a^y_{i,j+\frac12},\, b^y_{i,j+\frac12}$ defined in Section~\ref{subsec:rho}, we set
\[
\Delta z
=
\min\bigg\{
\Delta z_{\max},\ 
\mathrm{CFL}\cdot\bigg(
\frac{\max_{i,j} \big(a^x_{i+\frac12,j} + b^x_{i+\frac12,j} \big)}{\Delta x}
+
\frac{\max_{i,j} \big(a^y_{i,j+\frac12} + b^y_{i,j+\frac12}\big) }{\Delta y}
\bigg)^{-1}
\bigg\},
\]
where $\mathrm{CFL}>0$.

% ==========================================================
\section{Structure properties of the discrete scheme}
\label{sec:discrete_results}

We state two key properties of the fully discrete intensity update:
(i) exact conservation of discrete mass for each forward Euler stage, and
(ii) positivity preservation under a standard CFL condition for the FV/Rusanov
flux. These properties are inherited by SSP-RK3 because SSP-RK3 is a convex
combination of forward Euler steps.

% ----------------------------------------------------------
\begin{proposition}[Exact discrete mass conservation]\label{prop:mass_discrete}
Consider a forward Euler FV update for the intensity
\[
\rho^{+}_{i,j}
=\rho_{i,j}
-\Delta z\bigg(
\frac{F^x_{i+\frac12,j}-F^x_{i-\frac12,j}}{\Delta x}
+\frac{F^y_{i,j+\frac12}-F^y_{i,j-\frac12}}{\Delta y}
\bigg),
\]
with boundary face fluxes set to zero:
\[
F^x_{\frac12,j}=F^x_{N_x+\frac12,j}=0,\qquad
F^y_{i,\frac12}=F^y_{i,N_y+\frac12}=0.
\]
Then the discrete mass at step $n$ 
$
M_h(\rho):=\sum_{i=1}^{N_x}\sum_{j=1}^{N_y}\rho_{i,j}\,\Delta x\,\Delta y,
$
is preserved exactly: $M_h(\rho^+)=M_h(\rho)$.

\end{proposition}

\begin{proof}
Sum the update over all $(i,j)$ and multiply by $\Delta x\Delta y$:
\[
M_h(\rho^+)-M_h(\rho)
= -\Delta z\,\Delta y\sum_{j=1}^{N_y}\sum_{i=1}^{N_x}(F^x_{i+\frac12,j}-F^x_{i-\frac12,j})
-\Delta z\,\Delta x\sum_{i=1}^{N_x}\sum_{j=1}^{N_y}(F^y_{i,j+\frac12}-F^y_{i,j-\frac12}).
\]
Each double sum telescopes in its respective direction, leaving only boundary
fluxes, which vanish by assumption. Hence, the difference is zero.
\end{proof}
Next, we study the discrete centroid evolution. Define the centroid at step $n$ by
\[
y_{c,h} := \frac{1}{M_h}\sum_{i,j} y_j\,\rho_{i,j}\Delta x\Delta y,
\]
and similarly for $x_{c,h}$.

\begin{proposition}[Discrete centroid increment identity]
\label{prop:centroid_increment} Assume $M_h >0$.
For the Euler FV step
$
\rho^{+}_{i,j}
=\rho_{i,j}
-\Delta z\left(
\frac{F^x_{i+\frac12,j}-F^x_{i-\frac12,j}}{\Delta x}
+\frac{F^y_{i,j+\frac12}-F^y_{i,j-\frac12}}{\Delta y}
\right),
$
the centroid satisfies
$
y_{c,h}^{+}-y_{c,h}
=
\frac{\Delta z}{M_h}\sum_{i=1}^{N_x}\sum_{j=1}^{N_y-1}
F^y_{i,j+\frac12}\,\Delta x \Delta y,
$
with the convention that boundary face fluxes vanish. An analogous formula holds
for $x_{c,h}$ with $F^x$.
\end{proposition}

\begin{proof}
Multiply the FV update by $y_j$ and sum over $(i,j)$:
\[
\begin{aligned}
&\sum_{i,j} y_j(\rho^{+}_{i,j}-\rho_{i,j})\Delta x\Delta y\\
&\quad=
-\Delta z\sum_{i,j} y_j
\left(\frac{F^x_{i+\frac12,j}-F^x_{i-\frac12,j}}{\Delta x}\right)\Delta x\Delta y
-\Delta z\sum_{i,j} y_j
\left(\frac{F^y_{i,j+\frac12}-F^y_{i,j-\frac12}}{\Delta y}\right)\Delta x\Delta y.
\end{aligned}
\]
The $x$-flux sum telescopes to zero because $y_j$ is constant in $i$ and boundary
fluxes are zero. For the $y$-flux term, use summation by parts in $j$:
\[
-\sum_{j=1}^{N_y} y_j\frac{F_{j+\frac12}-F_{j-\frac12}}{\Delta y}\Delta y
=
\sum_{j=1}^{N_y-1}(y_{j+1}-y_j)F_{j+\frac12},
\]
and note $y_{j+1}-y_j=\Delta y$ for the uniform grid. This yields
\[
\sum_{i,j} y_j(\rho^{+}_{i,j}-\rho_{i,j})\Delta x\Delta y
=
\Delta z\sum_{i=1}^{N_x}\sum_{j=1}^{N_y-1}F^y_{i,j+\frac12}\,\Delta x \Delta y,
\]
because boundary contributions vanish by $F^y_{i,\frac12}=F^y_{i,N_y+\frac12}=0$.
Finally, divide by $M_h$ to obtain the centroid increment.
\end{proof}

% ----------------------------------------------------------
\begin{proposition}[Positivity of the FV/Rusanov intensity update]\label{prop:positivity}
Assume $\rho_{i,j}\ge 0$ for all cells, and that the FV fluxes are Rusanov (local LF) fluxes with wave-speed parameters chosen as
$
a^x_{i+\frac12,j} := |v^x_{i+\frac12,j}|,\qquad
a^y_{i,j+\frac12} := |v^y_{i,j+\frac12}|,
$
and the polarization gradient parameter 
$
b^x_{i+\frac12,j} := |\omega^x_{i+\frac12,j}|,\qquad
b^y_{i,j+\frac12} := |\omega^y_{i,j+\frac12}|.
$
If the $z-$step satisfies the global CFL condition
\begin{equation} \label{eq:CFL}
\Delta z\left(
\frac{\max_{i,j} c^x_{i+\frac12,j}}{\Delta x}
+
\frac{\max_{i,j} c^y_{i,j+\frac12}}{\Delta y}
\right)\le \frac{1}{2},
\end{equation}
where $c^x_{i+\frac12,j} = a^x_{i+\frac12,j}+b^x_{i+\frac12,j}$ ($y-$component similarily), then the forward Euler update yields $\rho^{+}_{i,j}\ge 0$ for all cells.
Consequently, SSP-RK3 preserves nonnegativity under the same CFL condition.
\end{proposition}

\begin{proof}
A fully discrete version of intensity equation writes:
\[
\frac{\rho^{+}_{i,j} - \rho_{i,j}}{\Delta z} = -\frac{F^x_{i+\frac12,j}-F^x_{i-\frac12,j}}{\Delta x} - \frac{F^y_{i,j+\frac12}-F^y_{i,j-\frac12}}{\Delta y}.
\]
Incorporating the form of the Rusanov flux, we arrive at
\[
\begin{aligned}
    \rho^{+}_{i,j} = \rho_{i,j} - \frac{\Delta z}{2\Delta x}\left((\rho_{i,j} + \rho_{i+1,j})(v^x_{i+\frac12,j} + \omega^x_{i+\frac12,j})- c^x_{i+\frac12,j} (\rho_{i+1,j} - \rho_{i,j}) \right) \\
    + \frac{\Delta z}{2\Delta x} \left((\rho_{i-1,j} + \rho_{i,j} )(v^x_{i-\frac12,j} + \omega^x_{i-\frac12,j}) - c^x_{i-\frac12,j} (\rho_{i,j} - \rho_{i-1,j}) \right) \\
   - \frac{\Delta z}{2\Delta y} \left((\rho_{i,j} + \rho_{i,j+1})(v^y_{i,j+\frac12} + \omega^y_{i,j+\frac12}) - c^y_{i,j+\frac12}(\rho_{i,j+1} - \rho_{i,j}) \right) \\
   +\frac{\Delta z}{2\Delta y} \left((\rho_{i,j-1} + \rho_{i,j})(v^y_{i,j-\frac12} + \omega^y_{i,j-\frac12})  - c^y_{i,j-\frac12}(\rho_{i,j} - \rho_{i,j-1}) \right).
\end{aligned}
\]
Since $c^x_{i+\frac12,j} = |v^x_{i+\frac12,j}| + |\omega^x_{i+\frac12,j}|$, this yields
\begin{subequations}\label{wavesp12}
\begin{align}
\frac{v^x_{i+\frac12,j} + \omega^x_{i+\frac12,j} + 
c^x_{i+\frac12,j}}{2} \leq \frac{|v^x_{i+\frac12,j}| + |\omega^x_{i+\frac12,j}| + 
c^x_{i+\frac12,j}}{2} \leq c^x_{i+\frac12,j}, \\
\text{and }\,\,\,\, \frac{c^x_{i+\frac12,j} - v^x_{i+\frac12,j} - \omega^x_{i+\frac12,j}}{2} 
\leq  \frac{a^x_{i+\frac12,j} + |v^x_{i+\frac12,j}| + |\omega^x_{i+\frac12,j}|}{2} \leq c^x_{i+\frac12,j}. 
\end{align}
\end{subequations}

Further rearranging the discretized intensity equation, we arrive at
\[
\begin{aligned}
\hspace{-0.5cm}    \rho^{+}_{i,j} = \rho_{i,j}\;-\;\rho_{i,j} \frac{\Delta z}{\Delta x}\left(\frac{v^x_{i+\frac12,j} + \omega^x_{i+\frac12,j} + c^x_{i+\frac12,j}} {2} \right) - \rho_{i,j} \frac{\Delta z}{\Delta x}\left(\frac{c^x_{i-\frac12,j} - v^x_{i-\frac12,j} - \omega^x_{i-\frac12,j}} {2} \right)\\
    - \rho_{i,j} \frac{\Delta z}{\Delta y}\left(\frac{v^y_{i,j+\frac12} + \omega^y_{i,j+\frac12} + c^y_{i,j+\frac12}} {2} \right) 
    - \rho_{i,j} \frac{\Delta z}{\Delta y}\left(\frac{c^y_{i,j-\frac12} - v^y_{i,j-\frac12} - \omega^y_{i,j-\frac12}} {2} \right) \\ 
    + \rho_{i+1,j} \underbrace{\frac{\Delta z}{\Delta x} \left(\frac{c^x_{i+\frac12,j} - v^x_{i+\frac12,j} - \omega^x_{i+\frac12,j}} {2} \right)}_{\textcolor{black}{w_1}} + \rho_{i-1,j} \underbrace{\frac{\Delta z}{\Delta x}\left(\frac{c^x_{i-\frac12,j} + v^x_{i-\frac12,j} + \omega^x_{i-\frac12,j}} {2}  \right)}_{\textcolor{black}{w_2}} \\
    + \rho_{i,j+1} \underbrace{\frac{\Delta z}{\Delta y}\left(\frac{c^y_{i,j+\frac12} - v^y_{i,j+\frac12} - \omega^y_{i,j+\frac12}} {2} \right)}_{\textcolor{black}{w_3}} + \rho_{i,j-1} \underbrace{\frac{\Delta z}{\Delta y}\left(\frac{c^y_{i,j-\frac12} + v^y_{i,j-\frac12} + \omega^y_{i,j-\frac12}} {2} \right)}_{\textcolor{black}{w_4}},
\end{aligned}
\]
\begin{equation}
\implies \rho^{+}_{i,j} =  w_0  \rho_{i,j} + w_1 \rho_{i+1,j} +  w_2 \rho_{i-1,j} + w_3 \rho_{i,j+1} + w_4\rho_{i,j-1}, \label{intensitydes1}
\end{equation} 
where 
\begin{align} \label{center_Weight}
w_0 = \frac12 - \frac{\Delta z}{\Delta x}\left(\frac{v^x_{i+\frac12,j} + \omega^x_{i+\frac12,j} + c^x_{i+\frac12,j}} {2} \right) - \frac{\Delta z}{\Delta y}\left(\frac{v^y_{i,j+\frac12} + \omega^y_{i,j+\frac12} + c^y_{i,j+\frac12}} {2} \right)\\
\,\,\,\,\,\,\, + \frac12 -  \frac{\Delta z}{\Delta x}\left(\frac{c^x_{i-\frac12,j} - v^x_{i-\frac12,j} - \omega^x_{i-\frac12,j}} {2} \right) -\frac{\Delta z}{\Delta y}\left(\frac{c^y_{i,j-\frac12} - v^y_{i,j-\frac12} - \omega^y_{i,j-\frac12}} {2} \right). \nonumber
 \end{align}
By defining the maxima of equation \eqref{wavesp12} over all vertical/horizontal faces, we can write
\[\begin{aligned}
\max_{i,j}\frac{v^x_{i+\frac12,j} + \omega^x_{i+\frac12,j} + 
c^x_{i+\frac12,j}}{2\Delta x}  \leq \max_{i,j} \frac{c^x_{i+\frac12,j}}{\Delta x},\\ 
\max_{i,j}\frac{v^y_{i,j+\frac12} + \omega^y_{i,j+\frac12} + 
c^y_{i,j+\frac12}}{2\Delta y}  \leq \max_{i,j} \frac{c^y_{i,j+\frac12}}{\Delta y}, \\
\max_{i,j}\frac{c^x_{i+\frac12,j} - v^x_{i+\frac12,j} - \omega^x_{i+\frac12,j}}{2\Delta x} 
\leq  \max_{i,j}\frac{c^x_{i+\frac12,j}}{\Delta x},\\
\max_{i,j}\frac{ 
c^y_{i,j+\frac12} - v^y_{i,j+\frac12} - \omega^y_{i,j+\frac12}}{2\Delta y}  \leq \max_{i,j} \frac{c^y_{i,j+\frac12}}{\Delta y} . \label{wavesp1}
\end{aligned}\]
Consequently, the global CFL condition \eqref{eq:CFL} implies the two bounds below:
\[\begin{aligned}
 0 \leq \frac12 - \Delta z \left(\max_{i,j}\frac{v^x_{i+\frac12,j} + \omega^x_{i+\frac12,j} + c^x_{i+\frac12,j}} {2\Delta x} + \max_{i,j} \frac{v^y_{i,j+\frac12} + \omega^y_{i,j+\frac12} + c^y_{i,j+\frac12}} {2\Delta y}\right),
\end{aligned}
\]
and
\[\begin{aligned}
0 \leq \frac12 - \Delta z \left(\max_{i,j}\frac{ c^x_{i-\frac12,j} - v^x_{i-\frac12,j} - \omega^x_{i-\frac12,j} } {2\Delta x} + \max_{i,j} \frac{c^y_{i,j-\frac12} - v^y_{i,j-\frac12} - \omega^y_{i,j-\frac12}} {2\Delta y}\right). 
\end{aligned}
\]
Summing the above two inequalities and using \eqref{center_Weight} yields $w_0 \geq 0$. Thus, we can see that expanding the Rusanov fluxes gives a forward Euler update of the form $\rho^+ = \sum_{\ell = 0}^{4} w_\ell \rho_\ell$, where the neighbor weights ($w_1, w_2, w_3, w_4$) are nonnegative with respect to the wave speed choice \eqref{wavesp12}, and the CFL condition guaranties that the remaining central weight $w_0$ is also nonnegative. Therefore, $\rho^+\ge 0$ whenever $\rho\ge 0$, since $\rho^+$ is a linear combination of nonnegative coefficients. Positivity preservation of finite volume schemes under suitable CFL restrictions is classical; see \cite{perthame1996positivity}. In our setting, the above argument yields the sufficient global positivity condition CFL$\,\leq \tfrac12$. Furthermore,
SSP-RK3 is a convex combination of forward Euler steps \cite{gottlieb2001strong}, so it preserves
positivity as well.
\end{proof}

\begin{proposition}[Monotonicity of the Godunov--LLF numerical Hamiltonian]
\label{prop:monotone_H}
Define the numerical Hamiltonian by
\[
H_{G\!L}(\phi)_{i,j}:=H_G(\phi)_{i,j}-\mathrm{diss}_{i,j},
\]
where $H_G$ is the Godunov numerical Hamiltonian \eqref{eq:GodunovHamiltonia} and $\mathrm{diss}_{i,j}$ is the LLF dissipation term \eqref{eq:diss}. Writing $H_{G\!L}(\phi)_{i,j} = \hat{H}(p_x^-, p_y^-, p_x^+, p_y^+)$ in terms of the one-sided slopes $p_x^\pm = (D_x^\pm \phi)_{i,j}, \,\, p_y^\pm = (D_y^\pm \phi){i,j}$, $\hat{H}$ is nondecreasing in $p_x^-, p_y^-$ and nonincreasing in $p_x^+, p_y^+$. Consequently, the forward Euler update for the H-J part:
\[
\phi^{n+1}_{i,j}=\phi^n_{i,j}-\Delta z\,H_{G\!L}(\phi^n)_{i,j} + \Delta z\,S^n_{i,j},
\]
where $S_{i,j}^n=S(\rho^n,\gamma^n)_{i,j}$ is defined by \eqref{eq:HJ_Hamiltonian} (and is independent of $\phi$); the update is monotone in the one-sided slopes whenever $\Delta z$ satisfies the corresponding CFL restriction.
\end{proposition}

\begin{proof}[Proof sketch]
For the convex Hamiltonian ${H}$ (in our case $H(p) = \frac12 |p|^2$), the Godunov numerical Hamiltonian $H_G$ is monotone in the one-sided slope arguments (nondecreasing in $p_x^-, p_y^-$ and nonincreasing in $p_x^+, p_y^+$); see \cite{osher1991high}. The Godunov construction for the H-J equation may also be interpreted via the local Riemann problem (see \cite{leveque2002finite}, sec.~12.5). 
The LLF term adds a discrete viscosity proportional to second differences, and the requirement that $\alpha_x,\alpha_y$ defined via \eqref{eq:charspeed}, dominate the characteristic speeds ensures that the combined operator $H_G-\mathrm{diss}$ remains monotone. Since $\Delta z\,S^n_{i,j}$ is independent of $\phi$, it enters additively and does not affect monotonicity in the one-sided slopes. Monotonicity of the Euler update then follows from standard scalar monotone-scheme arguments once $\Delta z$ satisfies the corresponding CFL
condition. 
\end{proof}

% ==========================================================
\section{Algorithm and Remarks}
\label{sec:algorithm}

This section summarizes the full solver in pseudocode and records practical remarks on stability, consistency, and diagnostics. All operators referenced below correspond to the discrete definitions given in
Section~\ref{sec:discretization}.

% ----------------------------------------------------------
\subsection{Algorithm (pseudocode)}
\label{subsec:pseudocode}

Algorithm~\ref{alg:ssprk3} states the SSP-RK3 solver for
\eqref{eq:FormI} with reflecting (Neumann) boundary conditions enforced via ghost cells and zero boundary flux.
All discrete operators referenced below are defined in
Section~\ref{sec:discretization}.

\medskip
\begin{center}
\begin{minipage}{0.95\linewidth}
\hrule\vspace{0.6em}
\textbf{Algorithm 1.} SSP-RK3 solver for \eqref{eq:FormI}
\label{alg:ssprk3}

\vspace{0.4em}
\textbf{Input:}
Grid parameters $(N_x,N_y,L)$; physical parameter $k_0$;
final distance $Z$; numerical parameters
$\rho_{\min}$, $s_{\min}$,\,$\mathrm{CFL}$, $\Delta z_{\max}$. Moreover: $\Delta x=2L/N_x$ and $\Delta y=2L/N_y$.

\vspace{0.3em}
\textbf{Initialize:}
Construct grid $(x_i,y_j)$.
Set initial fields $(\rho^0,\phi^0,\gamma^0)$.
Set $z\gets 0$.

\vspace{0.3em}
\textbf{While $z<Z$:}
\begin{enumerate}\setlength{\itemsep}{0.2em}

\item \textbf{CFL step size.}
Compute centered gradients $\nabla_h^0\phi, \,\nabla_h^0\gamma$ with Neumann ghosts and set
$\bm v=k_0^{-1}\nabla_h^0\phi$,\, $\bm \omega=k_0^{-1}\nabla_h^0\gamma$.
Choose
$
\Delta z=\min\!\bigg\{\Delta z_{\max},\;
\mathrm{CFL}\cdot \bigg(
\frac{\max_{i,j} \big(a^x_{i+\frac12,j} + b^x_{i+\frac12,j}\big)}{\Delta x}
+
\frac{\max_{i,j} \big(a^y_{i,j+\frac12}\;+ b^y_{i,j+\frac12}\big)}{\Delta y}
\bigg)^{-1}\bigg\}.
$

\item \textbf{Define RHS operator $\mathcal{R}(U)$.}
Given $U=(\rho,\phi,\gamma)$:
\begin{itemize}\setlength{\itemsep}{0.15em}
\item Compute 
$Q=\Delta_h\sqrt{\rho}/\max(\sqrt{\rho},s_{\min})$ with
$\rho\leftarrow\max(\rho,\rho_{\min})$.
\item Compute velocity $\bm v=k_0^{-1}\nabla_h^0\phi$ and polarization gradient $\bm \omega=k_0^{-1}\nabla_h^0\gamma$.
\item Intensity RHS:
$\mathcal{R}_\rho=-\Div_h^{\rm Rus}(\rho\bm v + \rho \bm \omega)$
(using zero boundary flux).
\item Phase RHS:
$
\mathcal{R}_\phi
=\frac{1}{k_0}\Big(\tfrac12 Q-\tfrac12|\nabla_h^0\gamma|^2\Big)
-\big(H_G(\phi)-\mathrm{diss}\big),
$
with $H_G$ the Godunov Hamiltonian and $\mathrm{diss}$ the LLF term.
\item Polarization RHS:
$\mathcal{R}_\gamma=-\,\bm v\cdot\nabla^{\rm up}\gamma$.
\end{itemize}

\item \textbf{SSP-RK3 update} using \eqref{eq:SSP-RK3}. 

\item \textbf{Stage sanitization.}
After each stage enforce
$\phi\gets\phi-\mathrm{mean}(\phi)$

\item \textbf{Advance.}
Set $z\gets z+\Delta z$.
\end{enumerate}

\vspace{0.35em}
\textbf{Output:}
Numerical solution $(\rho,\phi,\gamma)$ at $z=Z$ and  (optional) diagnostics.

\vspace{0.6em}\hrule
\end{minipage}
\end{center}

% ----------------------------------------------------------
\subsection{Consistency, stability, and diagnostic remarks}
\label{subsec:remarks}

\paragraph{Formal accuracy}
In regions where $(\rho,\phi,\gamma)$ remain smooth, the scheme is formally
observed to be second-order in $\phi$ and first order in $
\rho$ and $\gamma$. 
The Runge--Kutta integrator is third-order accurate in $z$.
Near steep gradients of $\phi$, monotone one-sided discretizations reduce local
order but guarantee stability and {convergence to the viscosity solution.}

\paragraph{Why LLF/Godunov is needed}
The phase equation is a nonlinear H-J equation.
Centered discretizations of $H(\nabla\phi)$ are non-monotone and can generate
spurious oscillations and catastrophic growth in $\|\nabla\phi\|$.
The Godunov numerical Hamiltonian provides a monotone approximation that
converges to the viscosity solution, while LLF dissipation damps unresolved
high-frequency gradients and improves robustness under long-distance
propagation and coupling to source terms.

\paragraph{CFL control}
We choose $\Delta z$ to satisfy the explicit CFL restriction associated with the FV transport speeds and the LLF/HJ dissipation parameters, and additionally enforce the sufficient FV positivity condition of proposition~\ref{prop:positivity}.
In practice, we also impose a maximum step size $\Delta z_{\max}$ to avoid overly large jumps when the beam is nearly stationary.

\paragraph{Mass conservation check}
Because the intensity equation is discretized in conservative form with zero
boundary fluxes, the discrete mass
is conserved up to $z-$ integration error whenever the positivity floor is inactive (otherwise the floor can introduce a small mass increase).
We monitor the relative deviation $|M(z)-M(0)|/M(0)$ as a primary diagnostic.

\paragraph{Gauge fixing for $\phi$}
The model depends on $\phi$ only through $\nabla\phi$ and $\partial_z\phi$,
so $\phi$ is determined up to an additive function of $z$.
Subtracting the spatial mean of $\phi$ after each stage fixes this gauge freedom and prevents accumulation of an arbitrary constant mode, improving numerical conditioning without affecting any physical quantity.

% ----------------------------------------------------------
\section{Numerical Experiments}
\label{sec:numerics}

This section presents numerical experiments illustrating the behavior of
system~\eqref{eq:FormI} and~\eqref{eq:Form_mod} under physically relevant initial data and compares the
computed beam dynamics with available analytical predictions and experimental results.
All simulations are performed using the fully specified numerical scheme
described in Sections~\ref{sec:model}–\ref{sec:algorithm}.

\subsection{No polarization case: Madelung system} In this subsection, we demonstrate the robustness of the scheme by applying it to the well-known Madelung equations, which describe a normally diffracting beam. Assuming the light beam has no polarization, the system \eqref{eq:FormI} reduces to the Madelung system \cite{madelung1927quantum}:
\begin{subequations}\label{eq:FormII}
\begin{align}
  \partial_z \rho
  + \frac{1}{k_0}\Div\big(\rho\,\grad \phi\big)
  &= 0,
  && \Omega \times (0,Z), \label{eq:rhoII}\\
  k_0\,\partial_z \phi
  + \frac12 |\grad\phi|^2
  &= \frac12\,Q(\rho),
  && \Omega \times (0,Z). \label{eq:phiII}
\end{align}
\end{subequations}
Fix $\sigma_0>0$ and define
\begin{equation}\label{eq:sigma-spread}
\sigma(z)^2:=\sigma_0^2+\frac{z^2}{4k_0^2\sigma_0^2}.
\end{equation}
Then the following Gaussian intensity profile, together with a quadratic phase augmented by an additional arctan term, gives an exact solution of \eqref{eq:FormII}:
\begin{subequations}\label{eq:FormII-solution}
\begin{align}
\rho(x,y,z) &=\frac{1}{2\pi\sigma(z)^2}\exp\!\left(-\frac{x^2+y^2}{2\sigma(z)^2}\right), \\
\phi(x,y,z)&=\frac{k_0\,z}{2\big(z^2+4k_0^2\sigma_0^4\big)}(x^2+y^2)
\;-\;\arctan\!\Big(\frac{z}{2k_0\sigma_0^2}\Big).
\end{align}
\end{subequations}
We note that the exact Gaussian solution \eqref{eq:FormII-solution} is naturally posed on $\mathbb{R}^2$ and, on a bounded
rectangular domain $\Omega=[-L,L]\times[-L,L]$, does not generally satisfy reflective (homogeneous Neumann)
boundary conditions \eqref{eq:BC_neumann}. Accordingly, for this verification test only, we impose the spatially and $z$-dependent Neumann data extracted from the exact solution and enforce it through ghost-cell reconstruction
(so that the discrete normal derivative matches the prescribed data). This follows the standard
finite-difference/finite-volume treatment of nonhomogeneous Neumann boundaries; see Sec.~10.6 of \cite{thomas2013numerical}. The parameters $k_0, \sigma_0$, CFL number, and domain half-width $L$ are the same as in Table~\ref{tab:params}.

\begin{figure}[h!]
    \centering
    \includegraphics[width=0.328\linewidth]{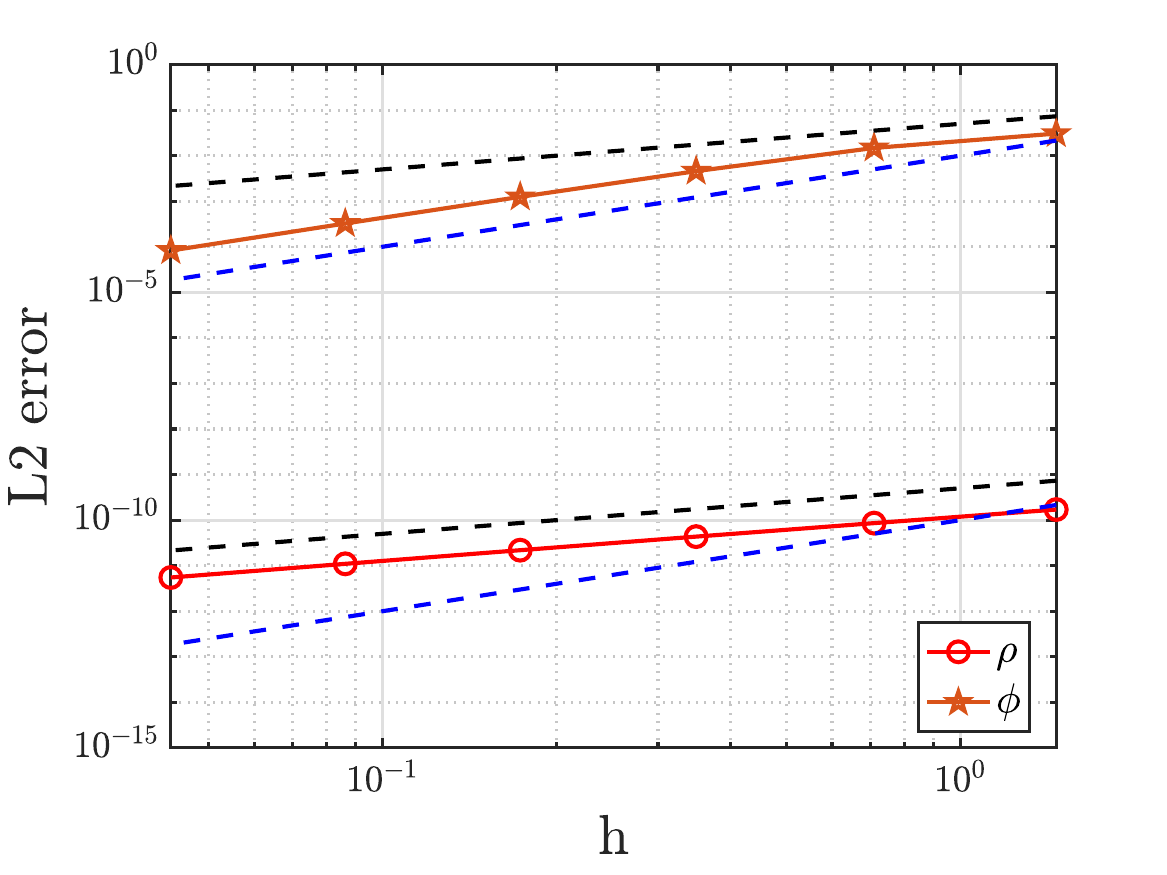}
     \includegraphics[width=0.327\linewidth]{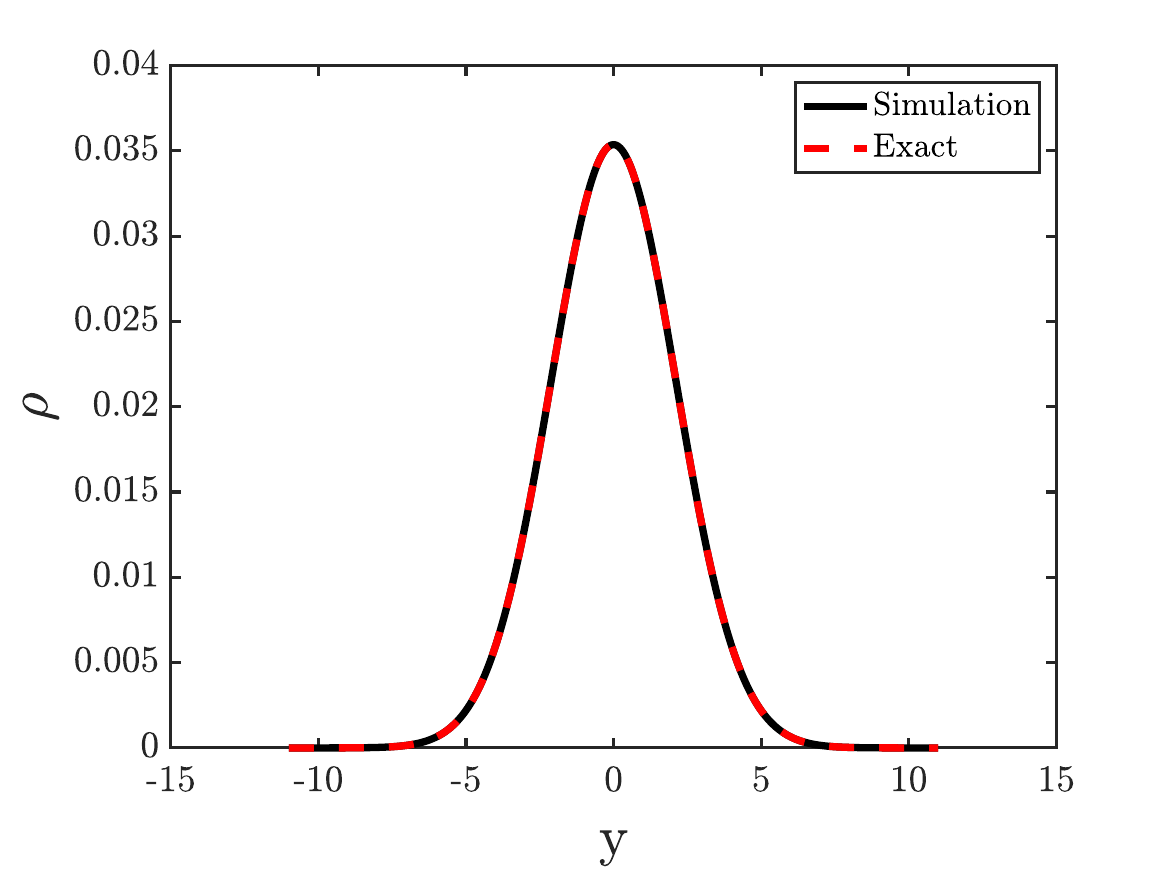}
    \includegraphics[width=0.327\linewidth]{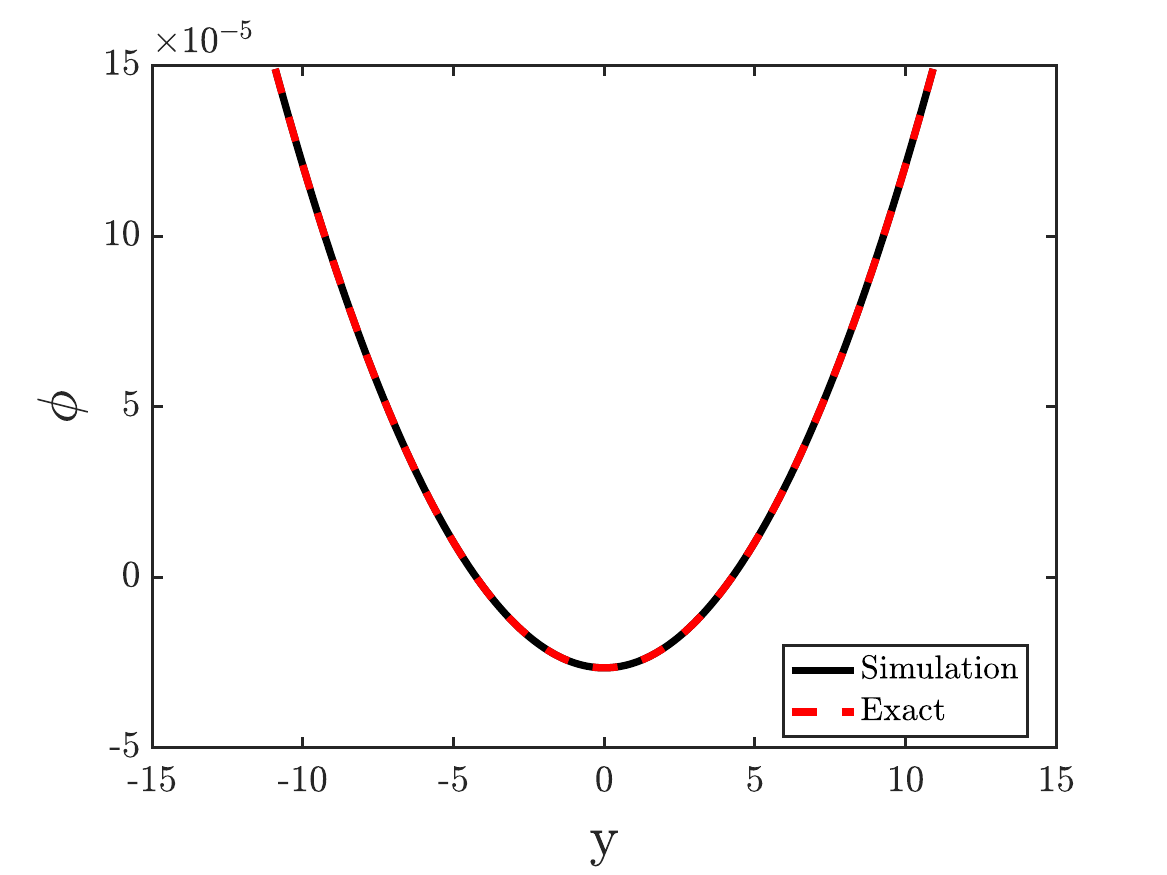}    
    \caption{Left panel: Log--log plot of the $L^2$-error versus mesh size $h$ for the uniform refinements
$N\in\{16,32,64,128,256,512\}$. The reference slopes indicate first-order convergence for $\rho$
(approximately parallel to the black line) and second-order convergence for $\phi$
(approximately parallel to the blue guideline). 
Middle panel: Comparison of the intensity $\rho$; the exact solution is shown by the red dotted curve
and the numerical approximation by the solid black curve.
Right panel: Comparison of the phase $\phi$.
Middle and right panels correspond to $x=0$ slice as a function of $y$ on the $512 \times 512$ mesh.}
    \label{fig:L2plot}
\end{figure}
Figure~\ref{fig:L2plot} (left) reports the error in the norm $L^\infty(0,Z;L^2(\Omega))$ under uniform mesh refinement
$N\in\{16,32,64,128,256,512\}$ (with $\Delta z$ chosen according to the CFL condition used throughout the paper).
We observe first-order convergence for $\rho$ (approximately parallel to the black reference slope) and
second-order convergence for $\phi$ (approximately parallel to the blue reference slope). Figure~\ref{fig:L2plot} (middle and right) overlays the exact and numerical solutions for $\rho$ and $\phi$
along the line $x=0$ (plotted as functions of $y$) at the final propagation distance
($Z=1$ mm, in this test case only), showing excellent agreement and confirming that the scheme reproduces
both the Gaussian profile and the quadratic phase.

% ----------------------------------------------------------
\subsection{Numerical setup}

Unless stated otherwise, all remaining simulations use the parameters listed in Table~\ref{tab:params}. The configuration corresponds to a paraxial optical beam propagating over distances of tens of meters, while the transverse domain remains on the millimeter scale. This pronounced scale separation is precisely the regime in which existing analytical results are valid only for short propagation distances, motivating the present large-scale numerical study. 

\begin{table}[h!]
\centering
\begin{tabular}{ll}
\hline
Parameter & Value \\ \hline
Domain half-width $L$ & $11\ \mathrm{mm}$ \\
Grid spacing $(\Delta x,\Delta y)$ & $\approx 0.06857\ \mathrm{mm}$ \\
Wavelength $\lambda$ & $1.5\times 10^{-3}\ \mathrm{mm}$ \\
Wavenumber $k_0$ & $2\pi/\lambda \approx 4189\ \mathrm{mm}^{-1}$ \\
Beam width $\sigma$ & $1.5 \times \sqrt{2}\ \mathrm{mm}$ \\
Polarization shift $x_0$ & $ -3.5, 3.5, 4.5, -5.5\ \mathrm{mm}$ \\
Final propagation distance $Z$ & $40\ \mathrm{m}$ \\
CFL number & $0.4$ \\
Maximum step size $\Delta z_{\max}$ & $10\ \mathrm{mm}$ \\
Intensity floor $\rho_{\min}$ & $10^{-20}\max\rho_0$ \\
\hline
\end{tabular}
\caption{Numerical parameters used in the simulations. Throughout these experiments, the profile scale is chosen as $a=|x_0|$.}
\label{tab:params}
\end{table} 

The step size in $z$ is chosen adaptively according to the CFL condition
described in Section~\ref{sec:algorithm}, with an additional cap
$\Delta z_{\max}$ to prevent overly large steps during early stages of the
propagation.

% ----------------------------------------------------------
\subsection{Initial conditions}
\label{s:inc}

The initial data consist of a centered Gaussian beam with a prescribed
polarization phase:
\begin{align*}
  \rho(x,y,0) = A^2 \exp\!\left(-\frac{x^2+y^2}{\sigma^2}\right),\quad 
  \phi(x,y,0) = 0,\quad 
  \gamma(x,y,0) = \frac{\pi}{2}\frac{(y-x_0)^2}{a^2} + \frac{\pi}{8}.
\end{align*}
The amplitude $A>0$ sets the initial peak intensity and does not influence the
centroid dynamics, which depend only on normalized moments of $\rho$. The parameters $a>0$ and $x_0$ determine the initial magnitude and direction of the scaled polarization gradient: $a$ sets the transverse length scale of the quadratic polarization profile, while $x_0$ specifies its offset relative to the beam centroid and hence the sign of the induced deflection. In the numerical experiments reported below, we choose $a=|x_0|$, consistent with the experimental parameter sets in \cite{Nichols:25a}. All simulations start from rest in the sense that $\nabla\phi(x,y,0)=0$, so that any transverse beam motion arises dynamically from coupling between
$\rho$, $\phi$, and $\gamma$.

\subsection{Beam centroid dynamics for the full model \eqref{eq:FormI}}
Figure~\ref{fig:centroid_phi+gamma} illustrates the centroid dynamics for the full model \eqref{eq:FormI}. Consistent with Theorem~\ref{thm:yc_linear}, the centroid exhibits linear-in-$z$ behavior, confirming that the full model does not produce the quadratic bending observed in the reduced model. This provides a useful baseline before turning to the reduced model \eqref{eq:Form_mod}, where momentum is no longer conserved and nonlinear centroid motion arises.

\begin{figure}[h!]
    \centering
    \includegraphics[width=0.49\linewidth]{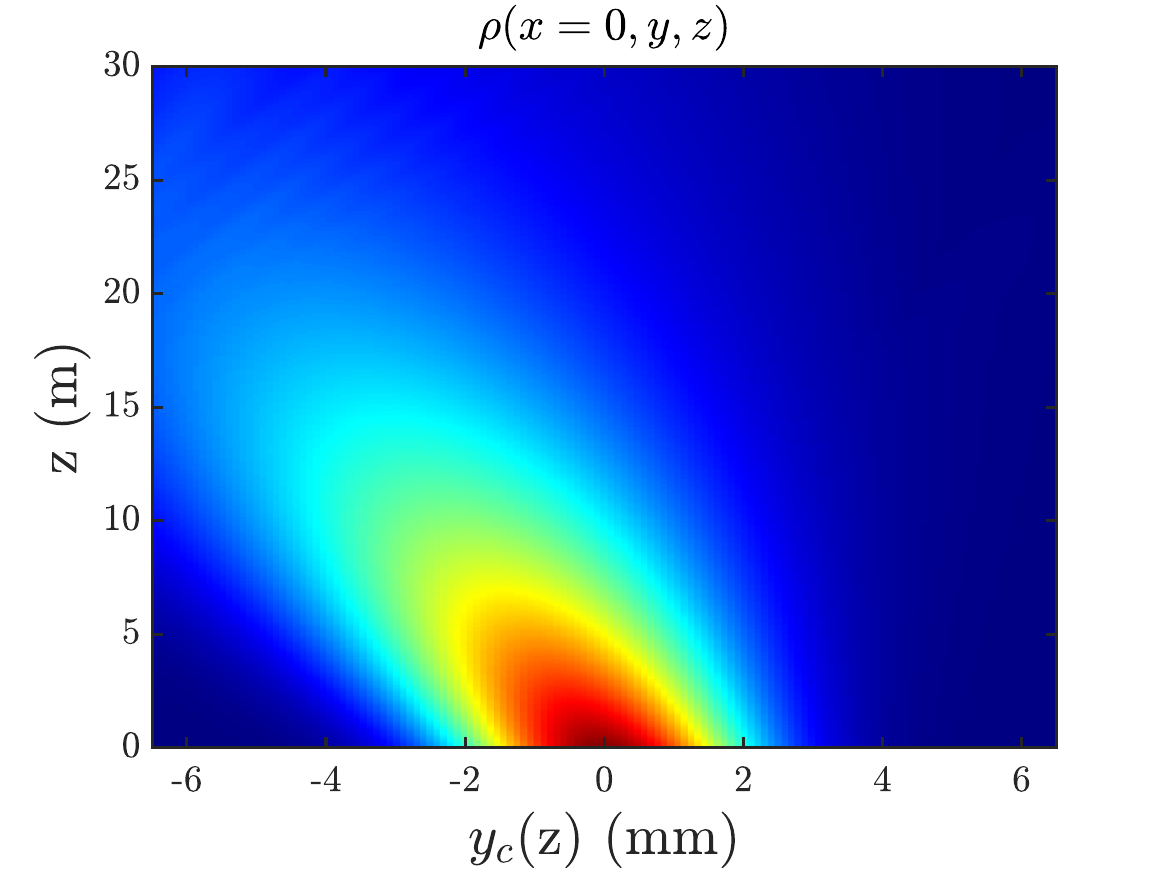}
    \includegraphics[width=0.49\linewidth]{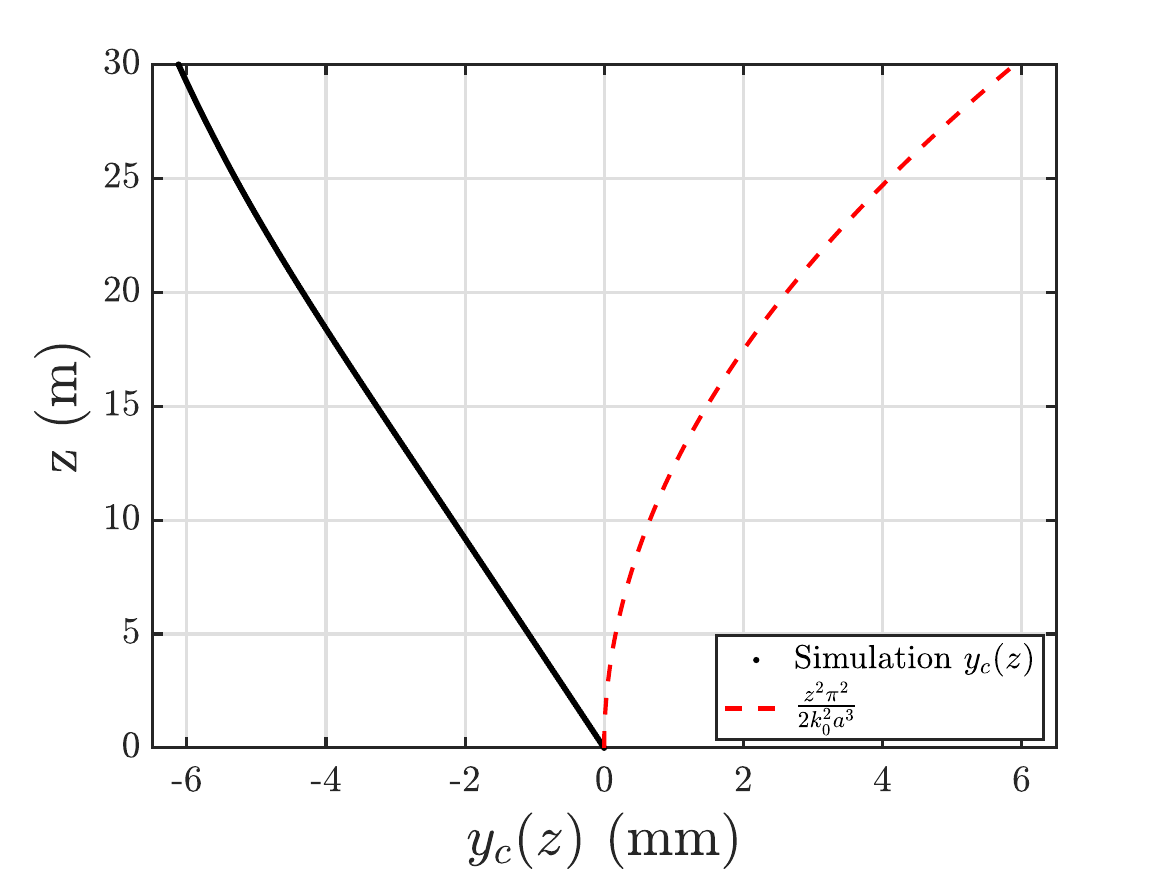}    
    \caption{Beam centroid dynamics for the full model \eqref{eq:FormI}. Left panel: Pseudocolor plot of the transverse intensity slice $\rho$ in the y-z plane, showing the linear drift of the beam during propagation. Right panel: Corresponding beam-centroid trajectory $y_c(z)$ as a function of propagation distance $z$. The simulation starts from $\nabla\phi_0\equiv 0,$ so the observed transverse motion is generated dynamically by the coupling between $\rho,\,\phi,$ and $\gamma.$}
    \label{fig:centroid_phi+gamma}
\end{figure}

% ----------------------------------------------------------
\subsection{Beam centroid dynamics and comparison with theory for reduced model \eqref{eq:Form_mod}}
The beam centroid is as given in \eqref{eq:centroid} and is computed numerically using discrete quadrature. For the Gaussian initial data considered here, the short-distance asymptotic prediction for the centroid is derived in Appendix~\ref{append:shortdistance}. In particular, the transverse displacement $y_c(z)$ follows a quadratic bending law for short propagation distances.
\begin{equation}\label{eq:theory_centroid}
  y_c^{\mathrm{th}}(z)
  =
  \frac{\pi^2}{2k_0^2 a^3}\,z^2 .
\end{equation}

For plotting, we report the propagation coordinate in meters.
When evaluating \eqref{eq:theory_centroid}, we convert $z$ to millimeters via
$z_{\mathrm{mm}} = 10^3 z_{\mathrm{m}}$ so that units remain consistent and
$y_c$ is reported in millimeters.

\begin{figure}[h!]
    \centering
    \includegraphics[width=0.4\linewidth]{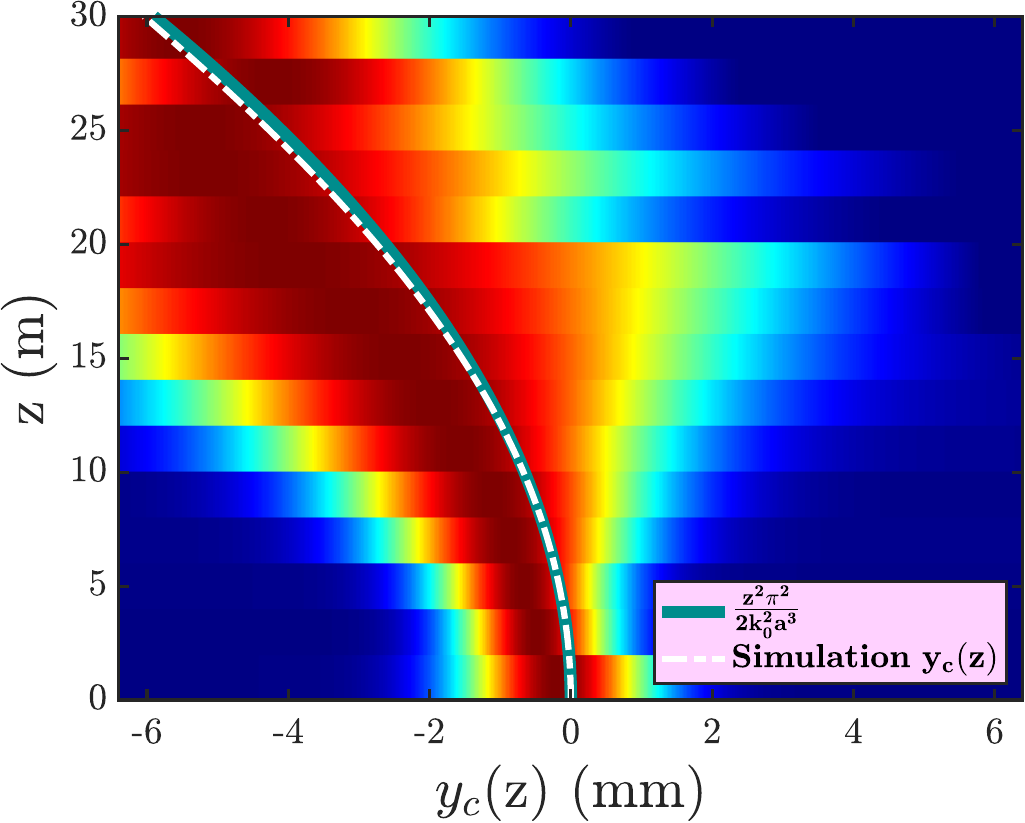}\quad
    \includegraphics[width=0.4\linewidth]{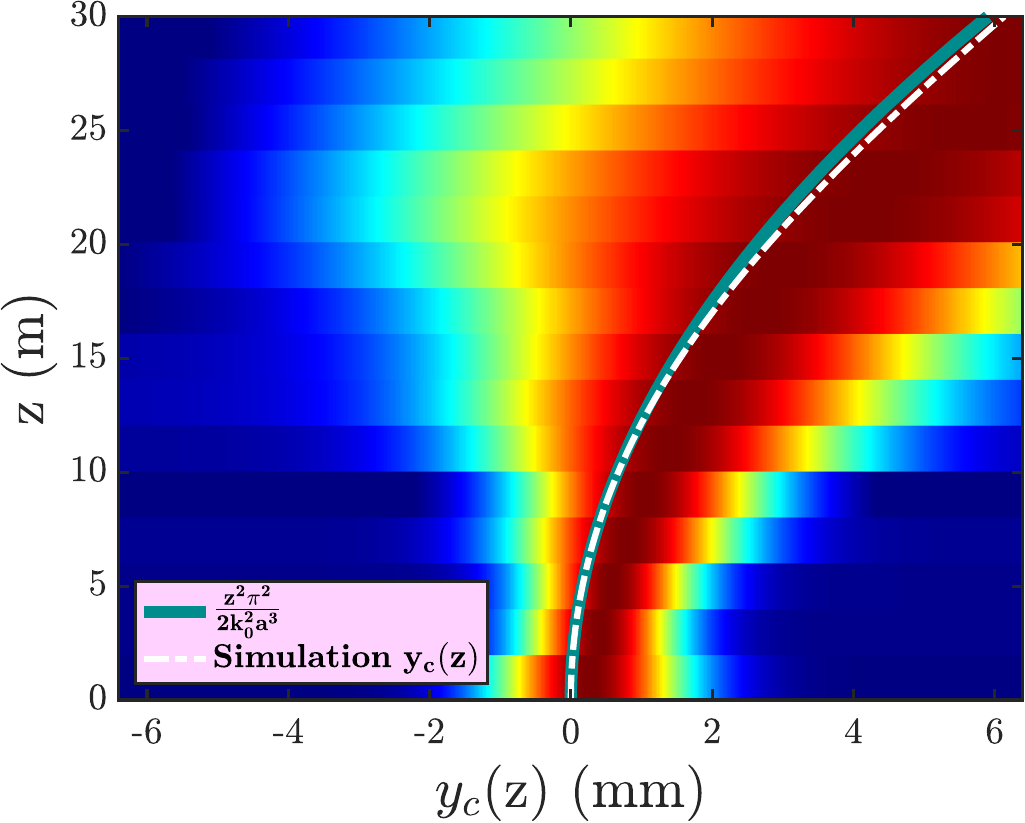}
    \includegraphics[width=0.4\linewidth]{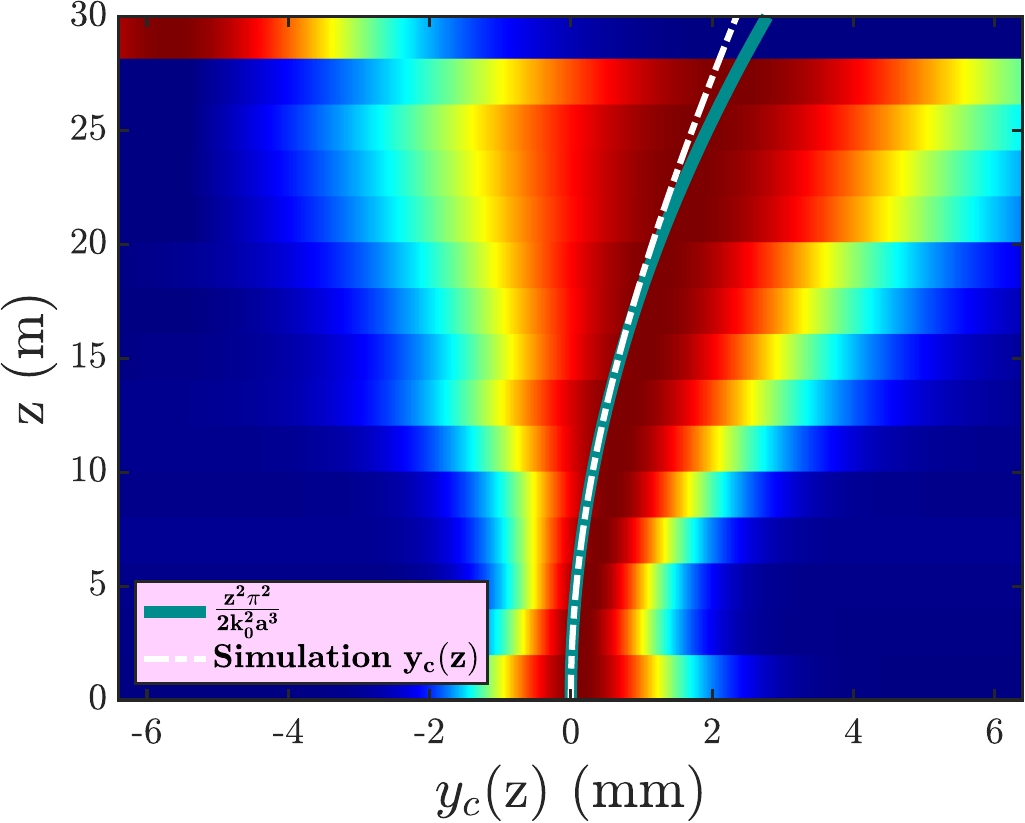} \quad
    \includegraphics[width=0.4\linewidth]{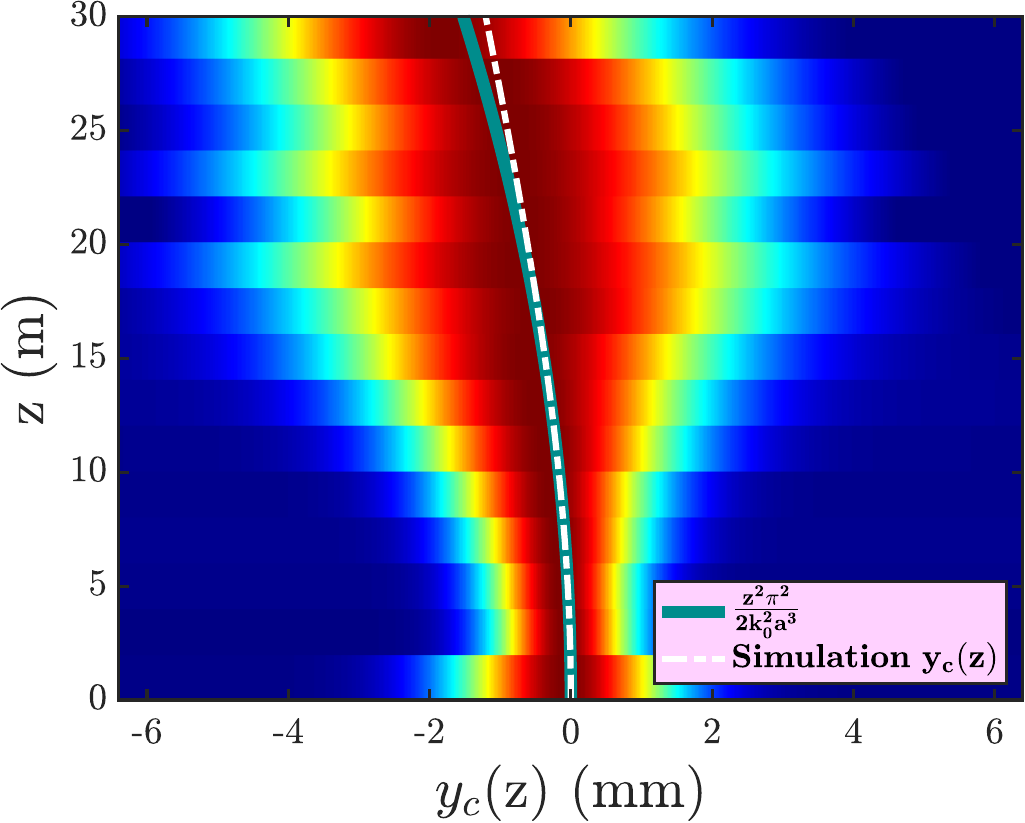}

    \caption{
Numerical vs analytical vs experimental (background) comparison. Top row: Here $x_0 = -3.5mm$ (left) and $x_0 = 3.5mm$ (right). Bottom row: Here $x_0 = 4.5mm$ (left) and $x_0 = -5.5mm$ (right). The white dotted lines represent our numerical simulation results, and the solid teal lines represent the short distance quadratic behavior from \eqref{eq:theory_centroid}. These results correspond to $Z = 30m$. 
}
    \label{fig:bending_80}
\end{figure}

Figure~\ref{fig:bending_80} compares the numerical centroid trajectory with the
analytical prediction~\eqref{eq:theory_centroid} and real experiment in the background. 
At short propagation distances (up to $30m$), the numerical solution closely follows the
quadratic law, confirming consistency with Appendix~\ref{append:shortdistance} and \cite{JMNichols_DVNickel_FBucholtz_2022a}.
At larger distances, systematic deviations from the quadratic behavior emerge, see Figure~\ref{fig:mass} (top row) where $Z = 40m$. The results are consistent as the spatial mesh is refined (not shown here). 
These deviations are expected and reflect the breakdown of the short-distance asymptotic assumptions.

\paragraph{Mass conservation}
Because the intensity equation is discretized conservatively with zero boundary
flux, the discrete mass $M_h(z)$ is conserved up to time-integration error
(Section~\ref{subsec:remarks}), see Figure~\ref{fig:mass} (bottom row).
In all runs reported here, $|M_h(z)-M_h(0)|/M_h(0)\le 10^{-12}$.

\begin{figure}[h!]
    \centering
    \includegraphics[width=0.4\linewidth]{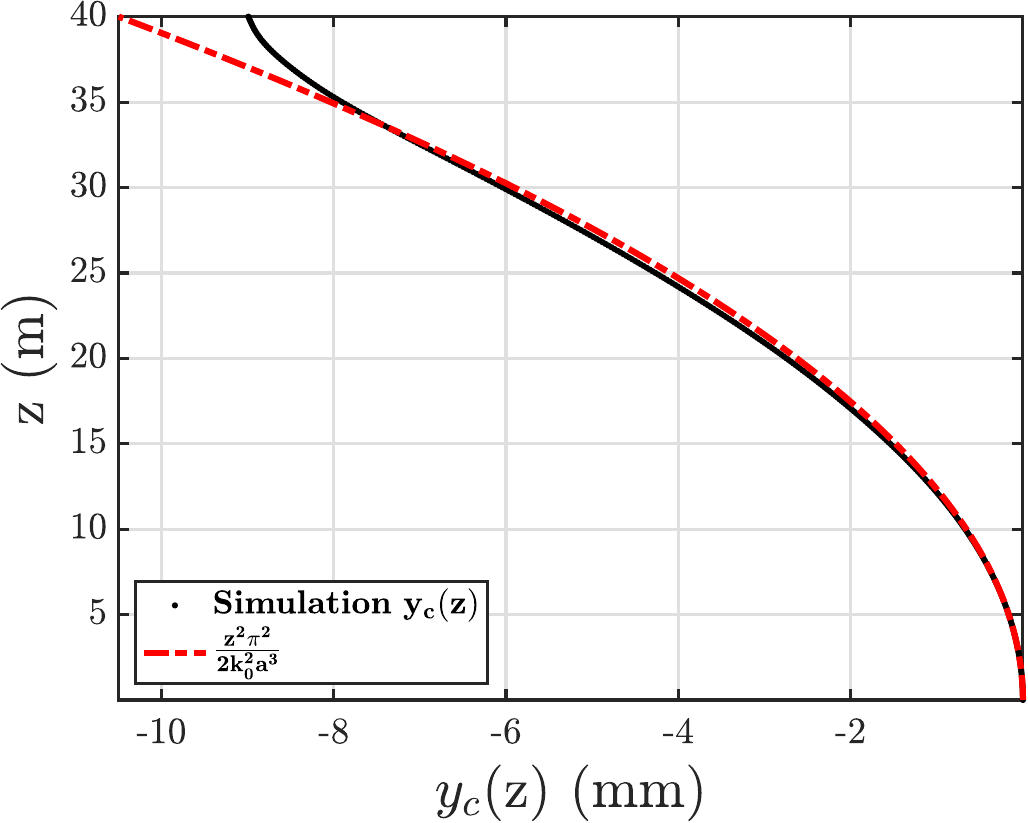}
    \includegraphics[width=0.4\linewidth]{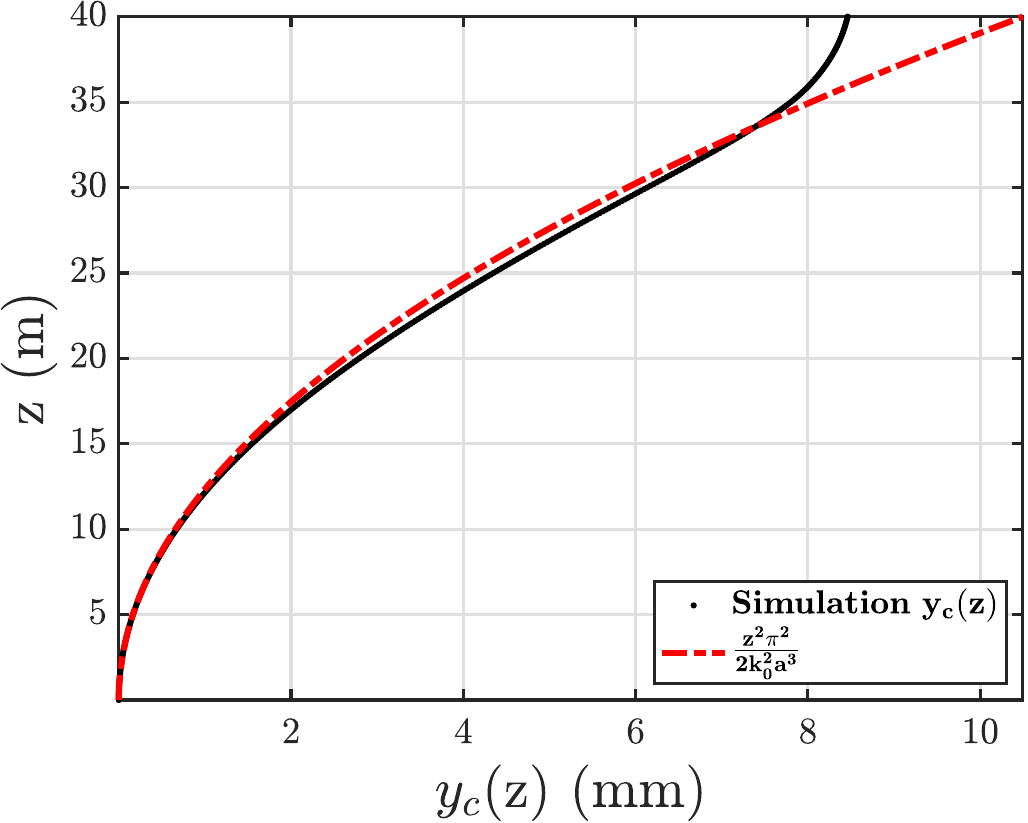}    
    \includegraphics[width=0.4\linewidth]{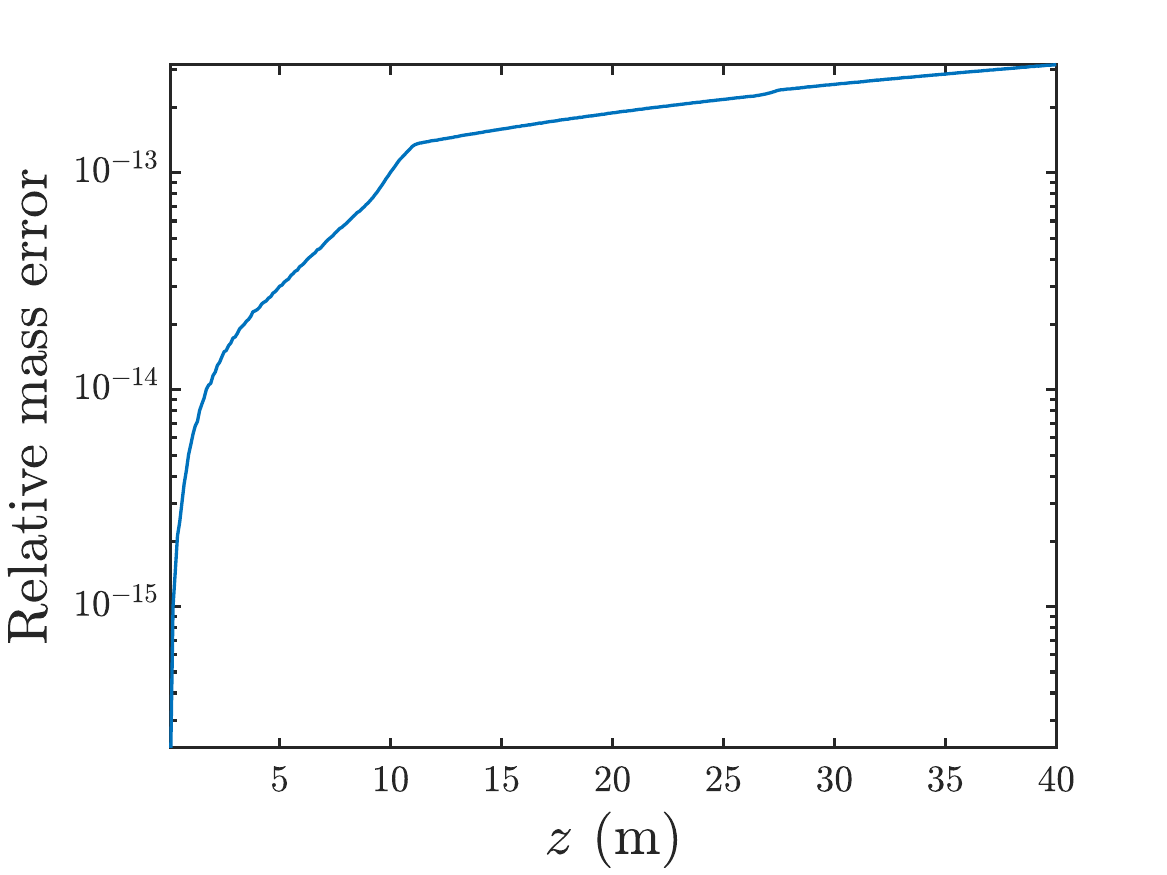}
    \includegraphics[width=0.4\linewidth]{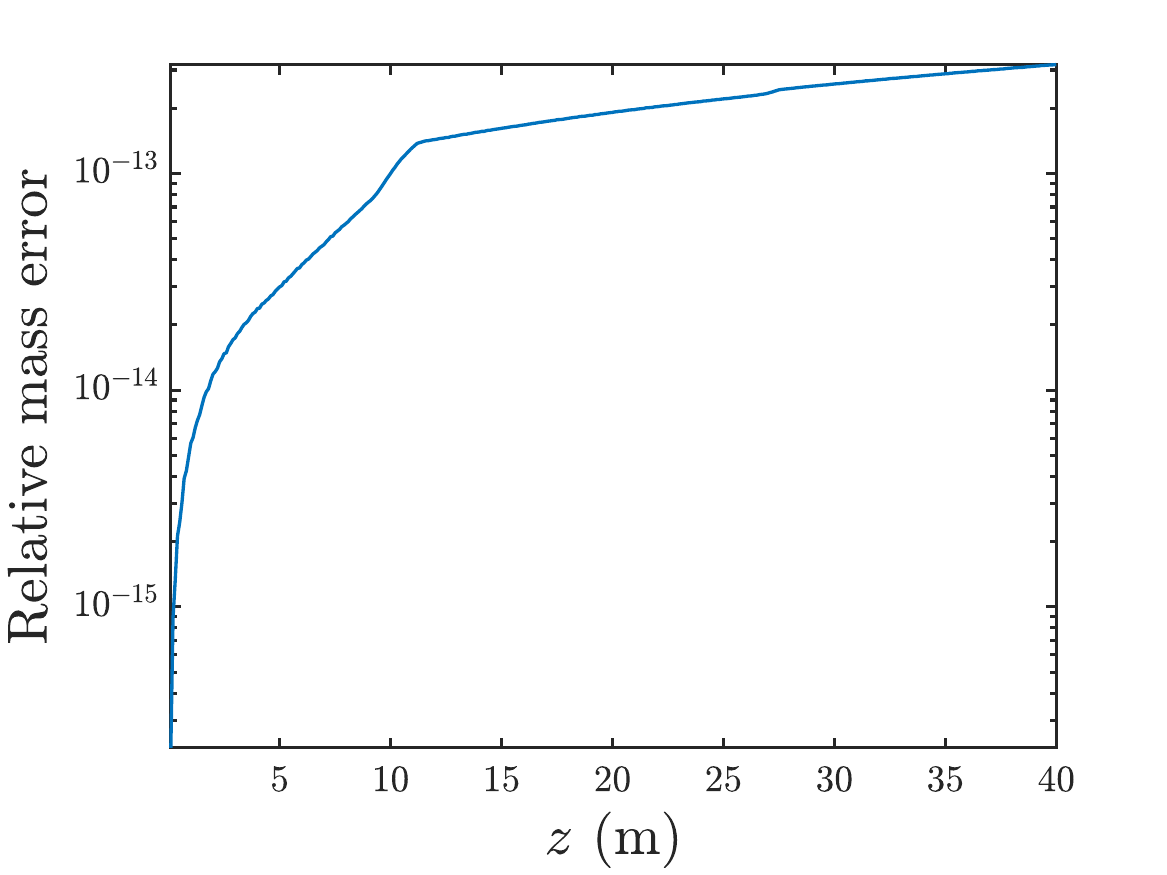}
    \caption{Top row: Numerical vs analytical comparison. Top row: Here $a = -3.5mm$ (left) and $a = 3.5mm$ (right). Here $Z = 40m$. 
    Bottom row:     
    Respective relative deviation of the total mass $M(z) = \int_\Omega \rho dx dy$. 
    The mass remains conserved, with relative deviations below $10^{-12}$ over the entire propagation distance confirming conservative behavior.}
    \label{fig:mass}
\end{figure}

Overall, these results validate the numerical scheme in the regime where analytical predictions apply and demonstrate the necessity of full-scale numerical simulation to access long-distance beam dynamics beyond the reach of short-distance asymptotic theory.

% ----------------------------------------------------------
\section{Discussion and Outlook}
\label{sec:outlook}

The numerical experiments presented in this work demonstrate that the model~\eqref{eq:Form_mod} captures subtle polarization-induced beam bending effects over propagation distances far beyond the reach of existing
analytical theory.
The agreement observed at short distances validates both the underlying model
and the numerical discretization, while the deviations at larger distances
highlight the limitations of short-distance asymptotic approximations.

From a modeling perspective, the results clarify the role of the transport
velocity in the intensity equation.
Transport by scalar phase gradient alone is essential for consistency with the analytical
bending law, whereas inclusion of the polarization phase gradient in this work produces qualitatively
different dynamics corresponding to a distinct physical model.

From a numerical standpoint, the study underscores the importance of monotone
H-J discretizations and conservative transport schemes when simulating coupled phase–intensity systems.
The combination of Godunov-type numerical Hamiltonians, local Lax--Friedrichs
stabilization, and conservative finite-volume fluxes proved essential for
robustness under long-distance propagation.\;Na\"ive centered discretizations were found to be unstable and unsuitable for this
class of problems.

Several extensions of this work are natural. First, a fully discrete entropy/energy stability theory and rigorous convergence (and, if feasible, long-time stability) results for the suggested discretization would be beneficial. Second, to minimize numerical diffusion while maintaining the monotonicity and robustness required for long-distance propagation, higher-order spatial discretizations may also be interesting.

\section*{Acknowledgment}
The authors are thankful to Drs. Jonathan Nichols and  Steven Rodriguez (U.S. Naval Research Laboratory, Washington D.C.) for providing several helpful suggestions. We are also grateful to Dr. Christopher Griffin (Penn State) for proof reading the manuscript and several helpful suggestions.

\appendix

%============================================================
\section{Quantum--stress identity and momentum balance}
\label{app:quantum}
%============================================================

%------------------------------------------------------------
\paragraph{Quantum stress tensor}
%------------------------------------------------------------
Let $\rho>0$ and set $w:=\sqrt{\rho}$ and
\[
Q(\rho):=\frac{\Delta w}{w}=\frac{\Delta \sqrt{\rho}}{\sqrt{\rho}}.
\]
Define the \emph{quantum stress} tensor
\begin{equation}\label{eq:quantum_stress_def}
%\boxed{
\bm S(\rho)
:=
-\frac{\rho}{4k_0^2}\,\nabla\otimes\nabla\big(\log\rho\big)
\;=\;
\frac{1}{2k_0^2}\Big(\nabla w\otimes \nabla w - w\,\nabla^2 w\Big).
%}
\end{equation}

\begin{lemma}\label{lem:Srelation_correct}
For $\rho>0$ (equivalently $w>0$) the following hold:
\begin{align}
%\boxed{
\frac{\Delta \sqrt{\rho}}{\sqrt{\rho}}
=
\frac{\Delta \rho}{2\rho}-\frac{|\nabla \rho|^2}{4\rho^2},
%}
\label{eq:Q_identity_rho}
\\[2mm]
%\boxed{
\Div\, \bm S(\rho)
=
-\frac{1}{2k_0^2}\,\rho\,\nabla\!\left(\frac{\Delta \sqrt{\rho}}{\sqrt{\rho}}\right)
=
-\frac{1}{2k_0^2}\,\rho\,\nabla Q(\rho).
%}
\label{eq:divS_identity}
\end{align}
\end{lemma}

\begin{proof}
\emph{Step 1: identity for $Q$.}
Write $w=\rho^{1/2}$. Then
\[
\nabla w=\frac{1}{2}\rho^{-1/2}\nabla\rho,
\qquad
\Delta w=\Div(\nabla w)
=\frac12 \rho^{-1/2}\Delta\rho-\frac14 \rho^{-3/2}|\nabla\rho|^2.
\]
Dividing by $w=\rho^{1/2}$ gives \eqref{eq:Q_identity_rho}.

\medskip
\emph{Step 2: divergence of $\bm S(\rho)$.}
Start from the $w$--representation in \eqref{eq:quantum_stress_def}:
\[
\bm S(\rho)=\frac{1}{2k_0^2}\big(\nabla w\otimes \nabla w - w\nabla^2 w\big).
\]
Use the standard vector/tensor identities (valid componentwise):
\[
\Div(\nabla w\otimes \nabla w)
=
(\nabla^2 w)\,\nabla w + (\Delta w)\,\nabla w,
\qquad
\Div\!\big(w\nabla^2 w\big)
=
w\,\nabla(\Delta w)+(\nabla^2 w)\,\nabla w.
\]
Subtracting yields
\[
\Div \bm S(\rho)
=
\frac{1}{2k_0^2}\Big((\Delta w)\nabla w - w\nabla(\Delta w)\Big)
=
-\frac{1}{2k_0^2}\,w^2\,\nabla\!\left(\frac{\Delta w}{w}\right).
\]
Since $w^2=\rho$ and $\Delta w/w=Q(\rho)$, this is exactly \eqref{eq:divS_identity}.
\end{proof}

\section{Recovering the short-distance quadratic law from Theorem 2.7}
\label{append:shortdistance}
We now show that Theorem~\ref{thm:yc_nonlinear} recovers the short-distance quadratic bending law
for the initial data used in Section~\ref{s:inc}. Recall from Theorem~2.7 that, for
the reduced model,
\[
y_c(z)
=
y_c(0)
+
\frac{P_y(0)}{M}z
-
\frac{1}{M}
\int_0^z\int_0^s F_y(\tau)\,d\tau\,ds
-
\frac{1}{k_0M}
\int_0^z
\left(
\int_\Omega \rho\,\gamma_y\,dxdy
\right)ds,
\]
where
\[
F_y(z)
=
\frac{1}{k_0^2}
\int_\Omega
\rho\,\nabla\gamma\cdot\nabla(\partial_y\theta)\,dxdy,
\qquad
\theta=\phi+\gamma .
\]
This is precisely the centroid identity established in Theorem~\ref{thm:yc_nonlinear} for the reduced model. 

Let the initial data be as in section~\ref{s:inc}. Then
\[
\gamma_{0,y}(y)=\frac{\pi}{a^2}(y-x_0),
\qquad
\gamma_{0,yy}=\frac{\pi}{a^2}.
\]
Since \(\rho_0\) is centered in \(y\), we have
\[
y_c(0)=0,
\qquad
\int_\Omega y\rho_0\,dxdy=0.
\]
Define
\[
G(z)
:=
\int_\Omega \rho(x,y,z)\gamma_y(x,y,z)\,dxdy.
\]
We first show that
\[
G'(0)=0.
\]
Differentiating gives
\[
G'(z)
=
\int_\Omega \rho_z\gamma_y\,dxdy
+
\int_\Omega \rho\,\partial_z\gamma_y\,dxdy.
\]
For the reduced model,
\[
\rho_z+\frac1{k_0}\operatorname{div}(\rho\nabla\phi)=0,
\qquad
\gamma_z+\frac1{k_0}\nabla\phi\cdot\nabla\gamma=0.
\]
Evaluating at \(z=0\), the assumption \(\nabla\phi_0\equiv0\) gives
\[
\rho_z(\cdot,0)
=
-\frac1{k_0}\operatorname{div}(\rho_0\nabla\phi_0)
=
0,
\]
and
\[
\gamma_z(\cdot,0)
=
-\frac1{k_0}\nabla\phi_0\cdot\nabla\gamma_0
=
0.
\]
Therefore
\[
\partial_z\gamma_y(\cdot,0)
=
\partial_y\gamma_z(\cdot,0)
=
0.
\]
Hence
\[
G'(0)
=
\int_\Omega \rho_z(\cdot,0)\gamma_{0,y}\,dxdy
+
\int_\Omega \rho_0\,\partial_z\gamma_y(\cdot,0)\,dxdy
=
0.
\]
Consequently,
\[
G(z)=G(0)+o(z)
\qquad
\text{as }z\downarrow0,
\]
and therefore
\[
\int_0^z G(s)\,ds
=
G(0)z+o(z^2).
\]

Next, since \(\nabla\phi_0\equiv0\), the initial total momentum satisfies
\[
P_y(0)
=
\frac1{k_0}
\int_\Omega \rho_0(\phi_{0,y}+\gamma_{0,y})\,dxdy
=
\frac1{k_0}
\int_\Omega \rho_0\gamma_{0,y}\,dxdy
=
\frac{G(0)}{k_0}.
\]
Thus the linear terms in Theorem~\ref{thm:yc_nonlinear} cancel:
\[
\frac{P_y(0)}{M}z
-
\frac{1}{k_0M}G(0)z
=
0.
\]
Since
\[
\int_0^z G(s)\,ds=G(0)z+o(z^2),
\]
the term involving \(G\) contributes no quadratic-order correction.

It remains to compute the leading contribution from \(F_y\). At \(z=0\),
\[
\theta_0=\phi_0+\gamma_0.
\]
Because \(\nabla\phi_0\equiv0\), we have
\[
\nabla(\partial_y\theta_0)
=
\nabla(\partial_y\gamma_0).
\]
Moreover, \(\gamma_0\) depends only on \(y\). Hence
\[
\nabla\gamma_0\cdot\nabla(\partial_y\theta_0)
=
\nabla\gamma_0\cdot\nabla(\gamma_{0,y})
=
\gamma_{0,y}\gamma_{0,yy}.
\]
Using the explicit formula for \(\gamma_0\),
\[
\gamma_{0,y}\gamma_{0,yy}
=
\frac{\pi}{a^2}(y-x_0)\frac{\pi}{a^2}
=
\frac{\pi^2}{a^4}(y-x_0).
\]
Therefore
\[
F_y(0)
=
\frac1{k_0^2}
\int_\Omega
\rho_0
\frac{\pi^2}{a^4}(y-x_0)
\,dxdy.
\]
Since \(\rho_0\) is centered in \(y\),
\[
\int_\Omega y\rho_0\,dxdy=0,
\qquad
M=\int_\Omega \rho_0\,dxdy.
\]
Thus
\[
F_y(0)
=
-\frac{\pi^2x_0}{k_0^2a^4}M.
\]
Assuming \(F_y\) is continuous at \(z=0\), we have
\[
\int_0^z\int_0^s F_y(\tau)\,d\tau\,ds
=
\frac12F_y(0)z^2+o(z^2) \qquad
\text{as }z\downarrow0.
\]
Substituting into the centroid identity gives
\[
y_c(z)
=
-\frac{1}{2M}F_y(0)z^2+o(z^2).
\]
Therefore
\[
y_c(z)
=
\frac{\pi^2x_0}{2k_0^2a^4}z^2
+
o(z^2)
\qquad
\text{as }z\downarrow0.
\]
In the special case \(x_0=a\), this becomes
\[
y_c(z)
=
\frac{\pi^2}{2k_0^2a^3}z^2
+
o(z^2).
\]
Thus Theorem~\ref{thm:yc_nonlinear} recovers the short-distance quadratic bending law.

\bibliographystyle{siamplain}
\bibliography{references}
\end{document}